\documentclass[final,1p,times]{elsarticle}

\usepackage{amssymb}

\usepackage[ruled, vlined, linesnumbered, norelsize]{algorithm2e}
\usepackage{amsthm}
\usepackage{amsfonts}
\usepackage{mathtools}
\usepackage{amsmath,enumitem} 
\usepackage{color,comment}
\usepackage{xcolor}
\usepackage{soul}
\usepackage{ulem}
\usepackage{cancel}
\usepackage{hyperref}
\newcommand{\eqbf}[1]{\boldsymbol{#1}}
\newtheorem{assumption}{Assumption}[section]
\newtheorem{definition}[assumption]{Definition}
\newtheorem{thm}[assumption]{Theorem}
\newtheorem{lemma}[assumption]{Lemma}

\newtheoremstyle{boldremark}
    {\dimexpr\topsep/2\relax} 
    {\dimexpr\topsep/2\relax} 
    {}          
    {}          
    {\bfseries} 
    {.}         
    {.5em}      
    {}          
\theoremstyle{boldremark}
\newtheorem{remark}[assumption]{Remark}

\newcommand{\model}{iGCF}
\usepackage{booktabs}

\newcommand{\eiOnline}{\eqbf{e}_{i}^*} 
\newcommand{\eiTOnline}{\eqbf{e}_{i}^{*\top}}

\newcommand{\euOnline}{\eqbf{e}_{u}} 
\newcommand{\eutOnline}{\eqbf{e}_{u,t}} 

\newcommand{\AlgPriorCov}{\eqbf{\Sigma}_0} 

\newcommand{\Emean}{\eqbf{\Phi}^\ast} 

\newcommand{\AlgPriorCovTau}{\eqbf{\Sigma}_{\tau + 1}} 
\newcommand{\noise}{\sigma_{noise}}

\journal{Artificial Intelligence}

\begin{document}

\begin{frontmatter}

\title{Interactive Graph Convolutional Filtering}

\author[sds]{Jin Zhang}
\ead{jinzhang21@mail.ustc.edu.cn}
\author[sds,cs,state]{Defu Lian}
\ead{liandefu@ustc.edu.cn}
\author[sds,cs]{Hong Xie}
\ead{xiehong2018@foxmail.com}
\author[bupt]{Yawen Li}
\ead{warmly0716@126.com}
\author[sds,cs,state]{Enhong Chen}
\ead{cheneh@ustc.edu.cn}
 \affiliation[sds]{organization={School of Data Science, University of Science and Technology of China},
             city={Hefei},
             postcode={230027},
             country={China}}

 \affiliation[cs]{organization={School of Computer Science and Technology, University of Science and Technology of China},
             city={Hefei},
             postcode={230027},
             country={China}}
 \affiliation[state]{organization={State Key Laboratory of Cognitive Intelligence},
             city={Hefei},
             postcode={230027},
             country={China}}
             
\affiliation[bupt]{organization={School of Economics and Management, Beijing University of Posts and Telecommunications},
             city={Beijing},
             postcode={100876},
             country={China}}

\begin{abstract}
Interactive Recommender Systems (IRS) have been increasingly used in various domains, including personalized article recommendation, social media, and online advertising. However, IRS faces significant challenges in providing accurate recommendations under limited observations, especially in the context of interactive collaborative filtering. 
These problems are exacerbated by the cold start problem and  data sparsity problem. 
Existing Multi-Armed Bandit methods, despite their carefully designed exploration strategies, often struggle to provide satisfactory results in the early stages due to the lack of interaction data. 
Furthermore, these methods are computationally intractable when applied to non-linear models, limiting their applicability. 
To address these challenges, we propose a novel method, the Interactive Graph Convolutional Filtering model. 
Our proposed method extends interactive collaborative filtering into the graph model to enhance the performance of collaborative filtering between users and items. 
We incorporate variational inference techniques to overcome the computational hurdles posed by non-linear models. 
Furthermore, we employ Bayesian meta-learning methods to effectively address the cold-start problem and derive theoretical regret bounds for our proposed method, ensuring a robust performance guarantee. 
Extensive experimental results on three real-world datasets validate our method and demonstrate its superiority over existing baselines.
\end{abstract}

\begin{keyword}
recommender systems \sep
interactive collaborative filtering \sep
graph model \sep
bandit \sep
meta-learning \sep
variational learning

\end{keyword}

\end{frontmatter}

\section{Introduction}
Over the past decade, interactive recommender systems (IRS) have received considerable attention due to their broad applicability in real-world scenarios \cite{zhao2013interactive, zhou2020interactive}, including personalized article recommendation \cite{li2010contextual}, social media \cite{guo2020deep}, and online advertising \cite{du2021exploration, wu2021adversarial}, among others \cite{wang2018online, wang2017factorization}. 
In contrast to traditional recommender systems \cite{koren2009matrix, he2020lightgcn, wang2019neural}, which treat recommendations as a one-step prediction task, IRS approach recommendations as a multi-step decision process. 
At each step, the system presents one or more items to the user and may receive feedback, which then sequentially influences subsequent recommendation decisions.
The system calculates rewards based on the feedback received, with the goal of maximizing cumulative rewards for a finite number of recommendations. 

The key challenge in IRS is to provide accurate suggestions for users under insufficient observations, especially for interactive collaborative filtering \cite{zhao2013interactive}. This context is characterized by the lack of feature representation for users and items, with available information limited to user and item IDs accompanied by user feedbacks for specific items.
Moreover, IRS faces significant problems with the cold start problem and data sparsity \cite{zhao2013interactive, zhou2020interactive, du2021exploration}.
The cold start problem occurs when new users enter the system without any interaction history, making it difficult to generate satisfactory recommendations. 
This problem is particularly pronounced in interactive collaborative filtering, where no additional features are available for these new users or items. 
In addition, data sparsity becomes a significant challenge when dealing with long-tail items. 
In many real-world datasets, item popularity often follows a power-law distribution \cite{lu2012recommender, newman2005power}, meaning that with the exception of popular items at the top of the distribution that receive significant attention, the majority of items have low exposure. 
The data sparsity of some items, i.e., they have only interacted with a small number of users, poses a significant challenge for 
interactive collaborative filtering (ICF) models to effectively learn the parameters of these items.
These items rarely receive enough user interactions, leaving the system with insufficient data to understand their characteristics and effectively recommend them, which affects the model's accuracy in predicting user behavior for items.

Efficient techniques to mitigate the aforementioned challenges require fast and accurate characterization of user profiles in as few interaction rounds as possible, while maintaining a degree of uncertainty in the prediction results. 
This requires a model that is capable of solving an exploration-exploitation (EE) dilemma, where it must balance the trade-off between its prediction results (exploitation) and uncertainty (exploration) in decision-making. 
In this context, bandit methods are particularly well suited to address such issues. 
In existing methods, Multi-Armed Bandit (MAB) methods conceptualize the recommendation task as multi-armed bandits or contextual bandits and address it with carefully designed exploration strategies such as Upper Confidence Bound (UCB) \cite{zhao2013interactive, li2010contextual} and Thompson Sampling (TS) \cite{chapelle2011empirical}.
However, existing bandit methods in ICF that rely on traditional matrix factorization, while providing some mitigations, struggle to address the problems caused by data sparsity.
In addition, when faced with the cold start problem, they have significant difficulty providing satisfactory results, often due to the lack of interaction data available in the early stages.
Furthermore, these methods prove to be computationally intractable when applied to non-linear models, severely limiting their applicability in the context of advanced deep models.

Motivated by the desire to overcome the limitations of existing bandit techniques, we present a novel method \textbf{\underline{I}}nteractive \textbf{\underline{G}}raph \textbf{\underline{C}}onvolutional \textbf{\underline{F}}iltering (\textbf{\model})  that effectively addresses these problems. 
The key features of our proposed methodology include the combination of bandit techniques with state-of-the-art graph neural networks, enabling our model to better exploit the power of collaborative filtering between users and items 
for overcoming the data sparsity problem, thereby significantly improving the model's expressiveness.
To overcome the computational hurdles posed by non-linear models, we also incorporate variational inference techniques to ensure effective computation even in the context of complex probabilistic models.  
Furthermore, we use meta-learning techniques to deal with the cold-start problem.

To summarize, the main contributions of this work are summarized as follows:
\begin{itemize}
    \item We propose a novel Interactive Graph Convolutional Filtering (\model) model, which extends the interactive collaborative filtering into the graph model and addresses the aforementioned shortcomings of existing bandit methods.  

    \item We employ a Bayesian meta-learning method to effectively deal with the cold-start problem, ensuring that our method maintains satisfactory recommendations even in the face of insufficient user-item interaction data.

    \item We derive theoretical regret bounds for our proposed method, providing a robust performance guarantee.
    
    \item We conduct extensive experiments on three real-world datasets to validate our method. The results consistently indicate that our method outperforms the state-of-the-art baselines, confirming the efficacy and applicability of our proposed methods.        
\end{itemize}

The rest of this paper is organized as follows. 
In the next Section \ref{sec:overview}, we present the notation used in this work and the overview of proposed method.
In Section \ref{sec:pre}, we present the pretraining phase of our proposed method, including the modeling approach and effective optimization techniques.
Then, in Section \ref{sec:o}, we discuss the online phase, detailing our online update strategy, operational considerations, and theoretical results. 
These include the derived regret bounds that provide a performance guarantee for our method.
In Section \ref{sec:exp}, we detail the experimental results, demonstrating the effectiveness of our method on three real-world datasets.
We then present a review of related work in Section \ref{sec:works}, discussing existing methods and their relation to our proposed method. 
Finally, in Section \ref{sec:con}, we conclude the paper with a summary.

\section{Model and Design Overview}
\label{sec:overview}

We first present a summary of key notations used throughout this paper.  
Then, we present recommendation model.   
Finally, we provide a comprehensive design overview of iGCF, outlining its key features.

\subsection{Notation}

We use the following notational conventions: 
bold lowercase and uppercase letters for vectors and matrices respectively, such as $\eqbf{a}$ and $\eqbf{A}$, and non-bold letters for scalars or constants, such as $k$ and $C$. The $l_2$ norm of a vector and the $l_2$ operator norm of a matrix are denoted by $\|\eqbf{a}\|$ and $\eqbf{\|A\|}_{\mathrm{op}}$, respectively. The smallest and largest eigenvalues of a matrix $\eqbf{A}$ are denoted by $\lambda_{\min}(\eqbf{A})$ and $\lambda_{\max}(\eqbf{A})$, respectively. The set $\{1, \ldots, n\}$, for any natural number $n$, is denoted by $[n]$. 
$\tilde{\mathcal{O}}$ symbolizes the $\mathcal{O}$ notation with polylogarithmic factors.
Also, some of the notations used in this paper are listed in Table \ref{tab:notation}.

\begin{table}[htb]
	\centering
	\caption{Notations and Definitions}\label{tab:notation}
	\begin{tabular}{l|l}
		\toprule 
		Notation & Definition \\
		\midrule $t$ & The round, $t\in [T] $.\\
		\midrule $M$ & The number of users.  \\
		\midrule $N$ & The number of items. \\
		\midrule $\mathcal{U}$ & The set consisting of $M$ users.  \\
        \midrule $\mathcal{I}$ & The set consisting of $N$ items.   \\
        \midrule $\mathcal{I}_{u,i}$ & The set of candidate items for recommendation to user \(u\) in round \(t\). \\
        \midrule $i_t \in \mathcal{I}_{u,i}$ & The item recommended to user $u$ in round $t$. \\
		\midrule $\eqbf{R} \in \mathbb{R}^{M \times N}$ & Observed rating or interaction matrix. \\
		\midrule $\eqbf{e}_u \in \mathbb{R}^d$ &  A vector representing user $u$. \\
		\midrule $\eqbf{e}_i \in \mathbb{R}^d$ &  A vector representing item $i$. \\  
              
		\bottomrule
	\end{tabular}
\end{table}

\subsection{Model}

Suppose the system has a set $\mathcal{U}$ of $M$ users  and a set $\mathcal{I}$ of $N$ items in the record. 
The user $u \in \mathcal{U}$ 
(or item $i \in \mathcal{I}$) is characterized by a feature or 
embedding vector $\eqbf{e}_u \in \mathbb{R}^d$ 
(or $\eqbf{e}_i \in \mathbb{R}^d$), 
where $d \in \mathbb{N}_+$. 
For the ease presentation, we denote the 
embedding matrix as $\eqbf{E}$, where 
\[
\eqbf{E} 
\triangleq 
[
(\eqbf{e}_u)_{u \in \mathcal{U}} 
\,\,\,\,
(\eqbf{e}_i)_{i \in \mathcal{I}}
].
\]
We consider the Bayesian setting that 
the embedding vector $\eqbf{e}_u$ (or $\eqbf{e}_i$) 
is drawn 
from a prior distribution $\mathcal{N}(\eqbf{\mu}_u, 
\eqbf{\Sigma}_u)$ with unknown 
$\eqbf{\mu}_u$ and $\eqbf{\Sigma}_u$ 
(or $\mathcal{N}(\eqbf{\mu}_i, 
\eqbf{\Sigma}_i)$ with unknown 
$\eqbf{\mu}_i$ and $\eqbf{\Sigma}_i$).

Let $r_{u,i}$ denotes a rating that 
quantifies the preference of user $u \in \mathcal{U}$ 
toward item $i \in \mathcal{I}$. 
The rating $r_{u,i}$ is a random variable.  
We consider two widely used rating models.  
One is the regression model: 
\[
r_{u,i} \sim  \mathcal{N} \left( \eqbf{e}_u^{\top} \eqbf{e}_i,\sigma_{{noise}}^2 \right),
\]
where $\sigma_{noise}^2$ characterizes the 
noise in observation. 
We also call this rating model the continuous 
feedback model.  
Another one is the Bernoulli model: 
\[
r_{u,i} 
\sim 
\text{Ber} (\sigma(\eqbf{e}_u^{\top} \eqbf{e}_i)),
\]
where $\text{Ber}(\cdot)$ denotes a Bernoulli distribution 
and  $ \sigma(x) = 1 / (1+\exp(-x)) $ denotes 
the sigmoid function. 
We also call this rating model the binary feedback model.  

Let $\delta_{u,i}$ denote an indicator such that 
$\delta_{u,i} = 1$ if and only if user 
$u$ interacts with item $i$ in the record, i.e., 
user $u$ assigns a rating to item $i$ in the record. 
We denote the interaction matrix as: 
\[
\eqbf{\Delta} 
\triangleq 
[\delta_{u,i}]_{u \in \mathcal{U}, i\in \mathcal{I}}.  
\]
Let $\eqbf{R} = [R_{u,i}]_{u \in \mathcal{U}, i\in \mathcal{I}}$ denote the feedback matrix in the record, 
where $R_{u,i} = r_{u,i}$ if $\delta_{u,i} =1$, 
otherwise $R_{u,i} = \text{null}$.  
Given the interaction matrix $\eqbf{\Delta}$ and 
the feedback matrix $\eqbf{R}$ in the record, 
our objective is to recommend items to 
users in an online manner. 
We consider the setting that new users 
may join the system. 
Without loss of generality, we focus on one user, 
i.e., user $u$, 
in delivering our method.  
Note that user $u$ can be either an existing user 
in the record or a new user. 
We aim to make $T\in \mathbb{N}_+$ rounds 
of recommendations to user $u$.    
Let $\mathcal{I}_{u,t}$ denote a set of 
candidate items to be recommended in round $t$. 
Each round recommends an item to user $u$.
Let $i_t \in \mathcal{I}_{u,t}$ denote the 
item recommended to user $u$ in round $t$. 
Our objective is to maximize the cumulative reward, 
which is defined as 
$\mathbb{E} [ \sum^T_{t=1} r_{u,i_t} ]$.  
Selecting $i_t$ to maximize the cumulative reward is technically 
nontrivial as both the user embedding vector $\eqbf{e}_u$ 
and item embedding vector $\eqbf{e}_i$ are unknown.  
Furthermore, the user $u$ may be a new user with no 
interaction history at all.  

\subsection{Design Overview of iGCF} 
We design iGCF, which utilize the interaction matrix $\eqbf{\Delta}$,  
the feedback matrix $\eqbf{R}$ and online feedback 
to select $i_t$.  
Overall, iGCF consists of pretraining and online phases, similar to other online methods \cite{zhao2013interactive}. 
\begin{itemize}
\item 
{\bf Pretrain phase of iGCF.} 
In the pretrain phase, we use historical interaction data to learn probability distributions for both users and items through a probabilistic graph-based recommendation model. 
We have developed models tailored for both continuous and binary feedback. 
To efficiently learn the posterior distributions of parameters in complex graph networks, we employ the variational inference approach. 
Specifically, we use the diagonal Gaussian distribution to approximate the posterior and the Monte Carlo method to optimize the variational lower bound.
The mean vectors obtained from the learned parameter distributions are then used in the online phase.

\item 
{\bf Online phase of iGCF.} 
In the online phase, we use the user vectors learned during pretrain to generate a meta-distribution for newly arriving users to ensure a positive initial interaction experience. Then,  using the item vectors obtained from the pretraining, we employ the Bayesian Linear UCB strategy to recommend items to users. The user's distribution is dynamically adjusted based on their real-time interaction data, resulting in a personalized recommendation.

\end{itemize}
We next present the details of the  pretrain phase of iGCF and online phase of iGCF individually.  

\section{Pretrain Phase of iGCF}
\label{sec:pre}

In this phase, we synergize the ideas of Probabilistic Matrix Factorization with LightGCN, a state-of-the-art graph-based recommendation model, to improve the embedding learning.
We first provide a basic overview of LightGCN to lay the foundation for the proposed method described below.
Then for each user and item vector, we place a prior distribution and obtain its posterior distribution through the graph model. 
We present two different modeling approaches: one corresponding to traditional regression and the other to binary classification. 
Given the computational challenges associated with the exact posterior distribution, we employ variational inference techniques, in particular using a diagonal Gaussian distribution to provide a tractable approximation of the posterior, and Monte Carlo sampling techniques for efficient optimization.
Finally, we discuss the possible extension of the proposed method to other graph models, emphasizing the scalability of our method.

\subsection{LightGCN}

We first review some basic elements of LightGCN  
that are helpful for delivering our iGCF pretrain 
method.  
LightGCN \cite{he2020lightgcn} treats both users and items as nodes in a bipartite graph, 
where each link represents an interaction 
between a user and an item.    
By performing multiple aggregations on the neighbors and adopting a streamlined network design, LightGCN achieves state-of-the-art performance in graph-based recommendation systems. 

Let 
$\boldsymbol{e}_u^{(0)}, \forall u \in \mathcal{U}$, and 
$\boldsymbol{e}_i^{(0)}, \forall i \in \mathcal{I}$, 
denote the initial embedding vectors for 
users and items respectively. 
Let $K \in \mathbb{N}_+$ denote the number of propagation layers of LightGCN.    
The graph convolution operation of LightGCN 
can be expressed as:
$$
\begin{aligned}
\eqbf{e}_u^{(k+1)} & =\sum_{i \in \mathcal{N}_u} \frac{1}{\sqrt{\left|\mathcal{N}_u\right|} \sqrt{\left|\mathcal{N}_i\right|}} \eqbf{e}_i^{(k)}, \\
\eqbf{e}_i^{(k+1)} & =\sum_{u \in \mathcal{N}_i} \frac{1}{\sqrt{\left|\mathcal{N}_i\right|} \sqrt{\left|\mathcal{N}_u\right|}} \eqbf{e}_u^{(k)},
\end{aligned}
$$
where 
$k \in \{0, 1, ..., K\}$, 
$\mathcal{N}_u$ denotes the set of items that 
user $u$ interacts with, $\mathcal{N}_i$ denotes the set of users that interact with item $i$.
The final embedding vectors denoted by  $\eqbf{\Bar{e}}_u$ and $\eqbf{\Bar{e}}_i$ 
are obtained through the graph readout operation: 
\[
\eqbf{\Bar{e}}_u=\sum_{k=0}^K \alpha_k \eqbf{e}_u^{(k)}, \quad \eqbf{\Bar{e}}_i=\sum_{k=0}^K \alpha_k \eqbf{e}_i^{(k)},
\]
where $\alpha_k \geq 0$ quantifies the importance of the $k$-th layer embedding in constituting the final embedding.  
It can be treated as a hyper-parameter.  
Let $\hat{r}_{u i}$ denote the predicted rating.  
The inner product of user and item final embedding vectors serves as the predicted rating:
$$
\hat{r}_{u i}=\eqbf{\Bar{e}}_u^{\top} \eqbf{\Bar{e}}_i,
$$
which is used as the ranking score for recommendation.

To improve the readability of this work, 
we review the matrix form of LightGCN. 
Recall the interaction matrix $\eqbf{\Delta}$.  
The adjacency matrix of the user-item graph 
can be expressed as: 
$$
\eqbf{A}=\left(\begin{array}{cc}
\eqbf{0} & \eqbf{\Delta} \\
\eqbf{\Delta}^\top & \eqbf{0}
\end{array}\right).
$$
Let the $k$-th layer embedding matrix 
and final embedding matrix be 
$
\eqbf{E}^{(k)} 
\triangleq 
[(\boldsymbol{e}^{(k)}_u)_{u \in \mathcal{U}} \,\, 
(\boldsymbol{e}^{(k)}_i)_{i \in \mathcal{I}}]
$ 
and 
$
\eqbf{\Bar{E}} 
\triangleq 
[(\eqbf{\Bar{e}}_u)_{u \in \mathcal{U}} \,\, 
(\eqbf{\Bar{e}}_i)_{i \in \mathcal{I}}]
$
respectively. 
The matrix equivalent form of graph convolution operation can be expressed as:
$$
\eqbf{E}^{(k+1)} = \eqbf{E}^{(k)} \left(\eqbf{D}^{-\frac{1}{2}} \eqbf{A} \eqbf{D}^{-\frac{1}{2}}\right),
$$
where $\eqbf{D}$ is a $(|\mathcal{U}|+|\mathcal{I}|) \times 
(|\mathcal{U}|+|\mathcal{I}|)$ diagonal matrix  with entry $D_{j j}$ denoting the number of nonzero entries in the $j$-th row vector of the adjacency matrix $\eqbf{A}$. 
Lastly, by setting 
\begin{equation}
    \eqbf{G} = \alpha_0 \eqbf{I} +\alpha_1 \tilde{\eqbf{A}} +\alpha_2 \tilde{\eqbf{A}}^2 +\ldots+\alpha_K \tilde{\eqbf{A}}^K,
\end{equation}
where $\tilde{\eqbf{A}}=\eqbf{D}^{-\frac{1}{2}} \eqbf{A D}^{-\frac{1}{2}}$ is the symmetrically normalized matrix,
we get the final embedding matrix used for model prediction as:

\begin{equation}
\label{eq:lightgcn_ma}
\begin{aligned}
\eqbf{\Bar{E}} 
& = \eqbf{E}^{(0)} \eqbf{G}. \\ 
\end{aligned}
\end{equation}
For the ease of presentation, 
we partition the column vectors of 
$\eqbf{G}$ into two groups, i.e., 
user group and item group, 
and index them accordingly such that 
\[
\eqbf{G} 
= 
[
(\boldsymbol{g}_u)_{u \in \mathcal{U}} 
\,\,
(\boldsymbol{g}_i)_{i \in \mathcal{I}}
].  
\]
Under this partition of $\eqbf{G}$, 
the formula of the final embedding can be 
rewritten as: 
\begin{equation}
\eqbf{\Bar{e}}_u 
= 
\eqbf{E} \eqbf{g}_u, 
\quad 
\eqbf{\Bar{e}}_i 
=
\eqbf{E} \eqbf{g}_i.  
\label{eq:LGCN:finalEmbG}
\end{equation}
 
\subsection{Interaction Matrix Aware Posterior 
Inference}

We utilize the interaction matrix 
$\eqbf{R}$ to pretrain the model  
borrowing the idea from LightGCN.  
In particular, we treat the 
embedding matrix $\eqbf{E}$ as the initial 
embedding matrix of the LightGCN model.   
We further use the corresponding final embedding vectors as the embedding vectors 
of each user and each item,  
i.e., 
$\eqbf{E} \eqbf{g}_u$ as the embedding vector of 
user $u$ and $\eqbf{E} \eqbf{g}_i$ as the 
embedding vector of item $i$ (following 
Equation (\ref{eq:LGCN:finalEmbG})).  
 
\noindent
{\bf Continuous feedback.} 
Replacing the feature vectors of users 
and items with those captured the interaction 
matrix $\boldsymbol{R}$, 
one can rewrite the continuous feedback model as: 
\begin{equation}
\label{eq::model::reg}
r_{u, i}= \eqbf{g}_u^{\top} 
\eqbf{E}^{\top} \eqbf{E} \eqbf{g}_i + \xi, \quad \xi \sim \mathcal{N}\left(0, \sigma_{noise}^2\right).
\end{equation} 
Note that both the $\eqbf{e}_u$ and $\eqbf{e}_i$ 
are unknown. 
To facilitate fast computation of the posterior, 
we place the following prior on the embedding vectors of users and items: 
\begin{equation}
\label{eq:para_init}
P(\eqbf{E})=\Pi_u \mathcal{N}\left( \eqbf{e}_u ; \eqbf{0}, \sigma_0^2 \boldsymbol{I}_{\boldsymbol{d}}\right) \cdot \Pi_i \mathcal{N}\left(\eqbf{e}_i; \eqbf{0}, \sigma_0^2 \boldsymbol{I}_d\right).     
\end{equation}
One can derive the posterior distribution of $\boldsymbol{E}$ under the above rating model as:
\begin{equation}
\label{eq:prob_loss_sq}
\begin{aligned}
P\left(\eqbf{E} \mid \eqbf{R}; \eqbf{G}\right) & \propto P\left(\eqbf{R} \mid \eqbf{E} ; \eqbf{G}\right) \cdot P\left( \eqbf{E} \right) \\
& \propto \Pi_{(u,i):\delta_{u, i}=1} \mathcal{N}\left(r_{u, i} ; \eqbf{g}_u^{\top} 
\eqbf{E}^{\top} \eqbf{E} \eqbf{g}_i, \sigma_{noise}^2\right) \cdot \Pi_u \mathcal{N}\left(\eqbf{e}_u ; \eqbf{0}, \sigma_0^2 \boldsymbol{I}_{\boldsymbol{d}}\right) 
\cdot \Pi_i \mathcal{N}\left(\eqbf{e}_i ; \eqbf{0}, \sigma_0^2 \boldsymbol{I}_{\boldsymbol{d}}\right), \\
\end{aligned}
\end{equation}
where $\delta_{u,i}$ denotes that user $u$ have interacted with item $i$, otherwise $0$.

\noindent
{\bf Binary feedback.} 
Replacing the feature vectors of users 
and items with those captured the interaction 
matrix $\boldsymbol{R}$, 
one can rewrite the binary feedback model as:   
\begin{equation}
\label{eq::model::cls}
r_{{u}, {i}}=
\text{Ber}
\left(\sigma\left( \eqbf{g}_u^{\top} 
\eqbf{E}^{\top} \eqbf{E} \eqbf{g}_i \right)\right). 
\end{equation}  
Similarly, one can derive the corresponding posterior distribution as: 

\begin{equation}
\label{eq:prob_sigmoid}
\begin{aligned}
P\left(\eqbf{E} \mid \eqbf{R}; \eqbf{G}\right) 
& \propto \Pi_{(u,i):\delta_{u, i}=1} 
\text{Ber}
\left(\sigma\left( 
\eqbf{g}_u^{\top} 
\eqbf{E}^{\top} \eqbf{E} \eqbf{g}_i
\right)\right)
\cdot \Pi_u \mathcal{N}\left(\eqbf{e}_u ; \eqbf{0}, \sigma_0^2 \boldsymbol{I}_{\boldsymbol{d}}\right) 
\cdot \Pi_i \mathcal{N}\left(\eqbf{e}_i ; \eqbf{0}, \sigma_0^2 \boldsymbol{I}_{\boldsymbol{d}}\right). \\
\end{aligned}
\end{equation}

\subsection{Variational Approximation of the Posterior}

Equation \eqref{eq:prob_loss_sq} and \eqref{eq:prob_sigmoid} 
demonstrate that 
that even with commonly used simple feedback models, the posterior distribution of user feature vectors and item feature vectors 
becomes non-Gaussian and computationally intractable due to the aggregation of neighbors in the graph network. 
To address this computational challenge, 
we employ the variational approximation method that uses a normal distribution to approximate the posterior distribution.  

Variational learning aims to find the parameters of a distribution on the feature matrix, denoted by $q(\eqbf{E})$, that minimizes the Kullback-Leibler (KL) divergence with the true Bayesian posterior distribution.  
We use a diagonal Gaussian distribution as the variational distribution: 
\begin{equation}
\label{eq:vi:dis}    
q(\eqbf{E}) = 
\prod_{u \in \mathcal{U}}
q\left(\eqbf{e}_u \right) 
\prod_{i \in \mathcal{I}} 
q\left(\eqbf{e}_i \right),
\end{equation}
where $q\left(\eqbf{e}_u \right)$ and 
$q\left(\eqbf{e}_i \right)$ follow 
Gaussian Distributions: 
\begin{equation}
    \label{eq:va}
q\left(\eqbf{e}_u \right) \sim \mathcal{N}\left(\boldsymbol{\mu}_{u}, \operatorname{Diag}\left(\boldsymbol{s}_{u}\right)\right), 
\quad 
q\left(\eqbf{e}_i \right) \sim \mathcal{N}\left(\boldsymbol{\mu}_{i}, \operatorname{Diag}\left(\boldsymbol{s}_{i}\right)\right).
\end{equation}
The KL divergence of $q$ from the posterior 
distribution $P$ of $\boldsymbol{E}$ can be 
derived as:
\begin{equation}
\label{eq:KL}
KL(q(\eqbf{E}) \| P(\eqbf{E} \mid R)) \propto \mathbb{E}_{\eqbf{E} \sim {q}}\left[\log \frac{q(\eqbf{E})}{P(\eqbf{E})}
- \log P(R|E)
\right].    
\end{equation}
Combining equations \eqref{eq:KL} and \eqref{eq:prob_loss_sq}, 
one can derive the loss function 
for the continuous feedback model as: 
\begin{equation}
    \label{eq::objective::reg}
    \begin{aligned}
\mathbb{E}_{\eqbf{E} \sim {q}}\left[{\mathcal{L}(\eqbf{E})} \right] = &\mathbb{E}_{\eqbf{E} \sim {q}} 
\Biggl[
\sum_{(u,i): \delta_{u i}=1} \frac{\left(r_{u, i} -  \eqbf{g}_u^{\top} \eqbf{E}^{\top}  \eqbf{E} \eqbf{g}_i\right)^2}{2 \sigma_{noise}^2} + \sum_{ u \in \mathcal{U} } 
\left( \frac{1}{2 \sigma_0^2}  
\eqbf{e}_u^{\top} \eqbf{e}_u - \frac{1}{2} \log \det( \operatorname{Diag}(\eqbf{s}_u))\right) \\
& + \sum_{ i \in \mathcal{I} } \left( \frac{1}{2 \sigma_0^2}  \eqbf{e}_i^{\top} \eqbf{e}_i - \frac{1}{2} \log \det( \operatorname{Diag}(\eqbf{s}_i))\right)
\Biggl] + const.
    \end{aligned}
\end{equation}

To effectively optimize the parameters of the variational distribution, 
one can obtain a sample of $\eqbf{e}_u$ 
by sampling from a standard multivariate normal distribution Gaussian distribution, shifting it by a mean 
$\boldsymbol{\mu}_u$, and scaling it by a standard deviation 
$\boldsymbol{s}_u$. 
Besides, to ensure that $\eqbf{s}_u$ is always non-negative, we parameterize the standard deviation pointwise as $\boldsymbol{s}_u = \log(1 + \exp(\boldsymbol{\rho}_u))$, which is  consistent with previous work \cite{blundell2015weight}.
Therefore, the transformation from a sample of parameter-free noise and the variational posterior parameters to obtain a posterior sample $\eqbf{e}_u$ is given by 
\begin{equation}
\label{eq:vi-optim}
\eqbf{e}_u = \boldsymbol{\mu}_u + \log(1 + \exp(\boldsymbol{\rho}_u)) \circ \boldsymbol{\epsilon}_u,
\end{equation}
where $\eqbf{\epsilon}_u \sim \mathcal{N}(\eqbf{0}, \eqbf{I_d})$ and $\circ$ denotes pointwise multiplication. 
One can have similar transformation for 
obtaining samples of $\eqbf{e}_i$.  

For the binary feedback model method \eqref{eq::model::cls}, we can use the same optimization techniques as described above, except that the loss function now takes the form:
\begin{equation}
    \label{eq::objective::cls}
    \begin{aligned}
\mathbb{E}_{\eqbf{E} \sim {q}}\left[{\mathcal{L}(\eqbf{E})} \right]  = &\mathbb{E}_{\eqbf{E} \sim {q}} 
\Biggl[
\sum_{(u,i): \delta_{u i}=1} -\log \sigma{\left(\left(2 r_{u, i} - 1\right) \eqbf{g}_u^{\top} \eqbf{E}^{\top} \eqbf{E} \eqbf{g}_i\right)}
+ \sum_{ u \in \mathcal{U} } 
\left( \frac{1}{2 \sigma_0^2}  
\eqbf{e}_u^{\top} \eqbf{e}_u - \frac{1}{2} \log \det( \operatorname{Diag}(\eqbf{s}_u))\right) \\
& + \sum_{ i \in \mathcal{I} } \left( \frac{1}{2 \sigma_0^2}  \eqbf{e}_i^{\top} \eqbf{e}_i - \frac{1}{2} \log \det( \operatorname{Diag}(\eqbf{s}_i))\right)
\Biggl] + const.
    \end{aligned}
\end{equation}
Other modeling approaches are also allowed, following the same optimization technique, requiring only modifications to the objective function.

Let $\boldsymbol{\mu}^\ast_u$, $\boldsymbol{s}^\ast_u$, 
$\boldsymbol{\mu}^\ast_i$ 
and $\boldsymbol{s}^\ast_i$ 
denote the optimal parameters 
obtained through the above optimization 
procedures. 
After the pretraining phase, 
the posterior distribution of user $u$ 
is approximated by the Gaussian distribution 
$
\mathcal{N}(\boldsymbol{\mu}^\ast_u, 
\operatorname{Diag}(\boldsymbol{s}^\ast_u)
)$ 
and the posterior distribution of the feature vector of item $i$ 
is approximated by the Gaussian distribution 
$
\mathcal{N}(\boldsymbol{\mu}^\ast_i, 
\operatorname{Diag}(\boldsymbol{s}^\ast_i)
)$. 
For the ease of presentation, 
we define the posterior mean matrix as 
\begin{equation}
\label{eq:pre_pos_mean}
\boldsymbol{\Phi^\ast} 
\triangleq 
[
(\boldsymbol{\mu}^\ast_u)_{u \in \mathcal{U}}
\,\, 
(\boldsymbol{\mu}^\ast_i)_{i \in \mathcal{I}}].  
\end{equation}

\begin{algorithm}[t]
    \KwIn{Training set $S = \{(u,i): \delta_{ui} = 1\}$, hyper-parameters $\sigma_0, \sigma_{noise}$, loss function $\mathcal{L}$, graph matrix $\eqbf{G}$, learning rate $\eta$.}
    $(\eqbf{\mu_u}, \eqbf{\rho_u}, \eqbf{\mu_i}, \eqbf{\rho_i}) \gets $ initialize all parameters with zeros for all users and items\;
    $\{S_b\}_{b=1}^B \gets$ split $S$ into several batches\;
	\For{b $\gets$ 1 to B}{
         $\eqbf{\epsilon}_u, \eqbf{\epsilon_i} \gets$ sample standard normal vectors from $\mathcal{N}(0, \eqbf{I})$ for each user and item\;
         $\eqbf{E} \gets$ construct embedding matrix using equation \eqref{eq:vi-optim}\;
         $\mathcal{\hat{L}} \gets$ compute loss $\mathcal{L}(\eqbf{E})$ using equation \eqref{eq::objective::reg} or \eqref{eq::objective::cls} on $S_b$\;
         $(\eqbf{\mu_u}, \eqbf{\rho_u}, \eqbf{\mu_i}, \eqbf{\rho_i}) \gets $ optimize parameters by SGD with learning rate $\eta$\;
        }
    $(\boldsymbol{\mu}^\ast_u, \boldsymbol{s}^\ast_u,\boldsymbol{\mu}^\ast_i, \boldsymbol{s}^\ast_i) \gets$  the optimal parameters, obtained by executing lines $2$-$7$ repeatedly until convergence\;
    $\boldsymbol{\Phi^\ast} \gets $  concatenate mean vectors of users and items using equation \eqref{eq:pre_pos_mean} \;
    \KwOut{$\boldsymbol{\Phi^\ast}$.}
 \caption{Pretrain Process of \model }\label{alg:pretrain} 
\end{algorithm}
We define the feature vectors of items under 
the optimal approximate posterior as 
\begin{equation}
\label{eq:pos_i}
\boldsymbol{e}^{\ast}_i 
\triangleq
\boldsymbol{\Phi^\ast} 
\boldsymbol{g}_i.  
\end{equation}
In the online phase of iGCF, 
we fix the feature vector of item $i$ 
to be $\boldsymbol{e}^{\ast}_i$. 
For clarity, we summarize the above process in Algorithm \ref{alg:pretrain}.
 
\subsection{Extending to Other Graph Models}
In this part, we explore the potential for extending our method to alternative graph models. 
Upon closer inspection, we find that the LightGCN model can be fully characterized by a convolutional coefficient matrix $\eqbf{G}$. 
Importantly, during the pretraining phase of our model, the representation and exploitation of the graph model is entirely dependent on this matrix \( \eqbf{G} \). 
This observation suggests that our framework is generalizable to any graph model that can be effectively represented by its respective convolutional coefficient matrix. 
For example, 
in SGCN \cite{wu2019simplifying}, the convolutional coefficient matrix can be represented as 
$$
\eqbf{G} = \left[(\eqbf{D} + \eqbf{I})^{-\frac{1}{2}} (\eqbf{A} + \eqbf{I}) (\eqbf{D} + \eqbf{I})^{-\frac{1}{2}} \right]^K,
$$
and in APPNP \cite{gasteiger2018predict}, the convolutional coefficient matrix can be represented as 
$$
\eqbf{G} = \beta \eqbf{I} + \beta(1 - \beta) \tilde{\eqbf{A}} + \beta(1-\beta)^2 \tilde{\eqbf{A}}^2 +\ldots + \beta(1-\beta)^K \tilde{\eqbf{A}}^K,
$$
where $\beta$ is the teleport probability to control the retaining of starting features in the propagation.
Our method can be easily adapted to above graph models. 
Even if other graph neural network models contain additional linear transformation layers, our approach remains applicable. 
This is provided that we do not consider the parameters of these linear layers as random variables, but rather as parameters to be optimized and include it as part of $\eqbf{G}$. 
Our method is well-suited for graph neural network architectures that lack non-linear activation functions. 
It should be noted that some research \cite{wu2019simplifying,he2020lightgcn} has explored the potential benefits of removing non-linear activation functions in graph models, particularly in the context of recommendation scenarios.

\section{Online Phase of iGCF}
\label{sec:o}
In this section, we first provide a basic overview of the ICF. 
We then discuss online aggregation in the context of newly arriving data. 
To address the challenges of the cold-start problem, we introduce a meta-learning method for rapid user initialization.
Next, we use Bayesian Linear UCB method to recommend items to users based on updated posterior distributions.
Finally, we provide a theoretical analysis of the regret associated with our proposed method, aiming to provide robust performance guarantees.

\subsection{ICF}
We will first provide some necessary background on ICF \cite{zhao2013interactive} which is helpful in understanding our iGCF online recommendation algorithm.
The techniques in ICF are based on probabilistic matrix factorization \cite{mnih2007probabilistic}. 
ICF utilizes the MCMC-Gibbs alternating optimization method to optimize the distributions of both users and items.
During online procedure, we can derive the posterior distributions of the users after the $(t-1)$-th interaction, denoting as 
$\mathcal{N}\left(\boldsymbol{e}_u; \boldsymbol{\mu}_{u, t}, \eqbf{\Sigma}_{u, t}\right)$, 
where the mean and variance terms can be obtained by following:
$$
\eqbf{\mu_{u,t}} = \left(\sum_{s=1}^{t-1} \eqbf{e}_{i_s} \eqbf{e}^\top_{i_s} + \lambda \eqbf{I} \right)^{-1} \left(\sum_{s=1}^{t-1} \eqbf{e}_{i_s} r_{u,i_{s}}\right),
\quad \eqbf{\Sigma}_{u, t}  = \left(\sum_{s=1}^{t-1} \eqbf{e}_{i_s} \eqbf{e}^\top_{i_s} + \lambda \eqbf{I} \right)^{-1} \sigma_{noise}^2,
$$
where $\eqbf{e}_{i_s}$ can be obtained by sampling from the item distribution or by using the maximizes the posterior probability (MAP) estimation of the item, and $\lambda$ is a regularization parameter.
Then, it selects the item for the $t$-th recommendation with the aim of maximizing the cumulative reward. Specifically, there are mainly two strategies have been explored to select the items in interactive collaborative filtering, i.e., upper confidence bound based method 
and Thompson sampling based method.  
Here we only introduce the upper confidence bound based method, 
since our method is built on it.  
It based on the principle of optimism in the face of uncertainty, which is to choose the item plausibly liked by users,
$$
i_t = \underset{i \in \mathcal{I}_{u,t}}{\arg \max } \ \left(\boldsymbol{\mu}_{u, t}^{\top} \boldsymbol{e}_{i} + c \sqrt{\log t}\left\|\boldsymbol{e}_{i}\right\|_{\eqbf{\Sigma}_{u, t}}\right),
$$
where $c$ is a constant to be tuned.

\subsection{Online Aggregation}
Before introducing the specifics of our method, let's first discuss the unique challenges associated with the arrival of new data in graph models. 
The main challenge in online updating of graph neural networks is that the new data changes the graph adjacency matrix, which in turn affects the global parameter.
Even with dynamic updating of the adjacency matrix, the complexity increases significantly with the depth of the graph neural network. 
Fortunately, in the online setting, the interaction data of a single user is relatively small compared to the existing interaction data (typically in the millions). 
It has minimal impact on the parameters of other users and items. 
Here, we follow the same approach as previous work \cite{zhao2013interactive} by only updating the parameters of the current user and not updating other users and items.

Importantly, due to the fixed distributions of other users and items, and under the premise of optimizing only the current user's distribution, we can use a single user distribution to replace the final distribution obtained by neighbor aggregation.
Without loss of generality, suppose $u$ is the current user.  
By the graph model, the feature vector of user $u$  can be 
derived as: 
$$
\eqbf{\Bar{e}}_{u} = \eqbf{E} \eqbf{g}_{u} = g_{u,u}\eqbf{e}_{u} + \sum_{u' \in \mathcal{U}\backslash{\{u\}}} g_{u,u'} \eqbf{e}_{u'} + \sum_{i \in \mathcal{I}} g_{u,i} \eqbf{e}_i.
$$
In our settings, the distribution of all embedding vectors 
remain fixed except $ \boldsymbol{e}_{u} $. 
The objective is to optimize $ \boldsymbol{e}_{u} $ so that $ \boldsymbol{\Bar{e}}_{u} \sim \mathcal{N}(\eqbf{\mu}_t, \eqbf{\Sigma}_t) $, where $ \mathcal{N}(\mu_t, \Sigma_t) $ denotes the posterior distribution for user $ u$ after observing data from $t-1$ interaction rounds, 
which will be computed explicitly in the following subsection.
It's worth noticing that even though the coefficient $ \eqbf{g}_{u} $ will change as more interaction data is observed, 
the desired posterior distribution we seek does not depend on $\eqbf{g}_{u} $. 
Instead, by adjusting $\eqbf{e}_{u'}$, we can establish an equilibrium that ensures that the distribution of $ \eqbf{\Bar{e}}_{u} $ matches our desired posterior distribution.
This implies that we can bypass the complicated process of neighbor aggregation. 
Instead, we can use a single user's parameter distribution $ \eqbf{e}_u $, to replace the final user distribution $ \Bar{\eqbf{e}}_u $ which is derived by neighbor aggregation. It is possible since the variational distributions of different users and items are independent, as represented by the equation \eqref{eq:vi:dis}. 
In the following discussion, we will focus our attention on $\eqbf{e}_u$ instead of $\eqbf{\Bar{e}}_u$, and we will use $\eqbf{e}_u$ as a replacement for $\eqbf{\Bar{e}}_u$.

\subsection{Meta Distribution and Posterior Update}

The pretraining phase of iGCF 
generates a large number of posterior  distributions of user feature vectors. 
Using existed user distributions, we can perform fast initialization for a newly arrived user, ensuring that the new user can have a good recommendation experience in the early stage. Concretely, 
$$
\eqbf{\mu}_\text{meta} 
= 
\frac{1}{ |\mathcal{U}| } 
\sum_{u \in \mathcal{U}} 
\eqbf{\Phi}^\ast \eqbf{g}_u, \quad
\eqbf{\Sigma}_\text{{meta}} 
=  
\frac{1}{|\mathcal{U}|-1} 
\sum_{u \in \mathcal{U}} 
(\eqbf{\Phi}^\ast \eqbf{g}_u - \eqbf{{\mu}_\text{meta}}) 
( \eqbf{\Phi}^\ast \eqbf{g}_u - \eqbf{{\mu}_\text{meta}})^{\top}.  
$$

We have $P(u_{new}) \sim \mathcal{N}(\eqbf{\mu}_\text{meta}, \eqbf{\Sigma}_\text{{meta}})$ as new user initial distribution.
At round $t$, the system recommends item to the user based on the user's posterior distribution, represent $\eqbf{{e}}^*_{i_t} \in \mathbb{R}^{d}$.
Then, the systems get user feedback ${r_t}$, 
which is a sample of $r_{u, i_t}$.   
Based on the new feedback, 
we update the posterior distribution of the user 
feature vector to be $\mathcal{N}(\eqbf{\mu}_{t}, \eqbf{\Sigma}_{t})$ with  
\begin{equation}
\begin{gathered} 
\eqbf{\Sigma}_{{t}}^{-1} = \eqbf{\Sigma}_{t-1}^{-1} + \frac{1}{\sigma_{noise}^2} \eqbf{e}^\ast_{i_t} \eqbf{e}^{\ast \top}_{i_t}, \\
\eqbf{\mu}_{t} = \eqbf{\Sigma}_{t}\left[  \eqbf{\Sigma}_{t-1}^{-1} \mu_{t-1} + \frac{1}{\sigma_{noise}^2} {r_t} \eqbf{e}^\ast_{i_t} \right],
\end{gathered}
\label{eq:update_dis}
\end{equation}
where $\eqbf{\Sigma}_{0} = \eqbf{\Sigma}_\text{{meta}} + \gamma \cdot \eqbf{I}$,  $\eqbf{\mu}_{0} = \eqbf{\mu}_\text{meta} $, and $\gamma$ is a hyperparameter to be tuned.
To make recommendations for users who appeared during the pretrain phase, we make only the following adjustments to the initial distribution.
\begin{equation}
    \label{eq:init:pre}
\begin{gathered}
\eqbf{\Sigma}_{{0}}^{-1} = \left(\eqbf{\Sigma}_\text{{meta}} + \gamma \cdot \eqbf{I}\right)^{-1} + \frac{1}{\sigma_{noise}^2} \eqbf{X}^{\top}_{0} \eqbf{X}_0, \\
\eqbf{\mu}_{0}=\eqbf{\Sigma}_{0}\left[ \left(\eqbf{\Sigma}_\text{{meta}} + \gamma \cdot \eqbf{I}\right)^{-1}  \eqbf{\mu}_\text{meta} + 
\frac{1}{\sigma_{noise}^2} \eqbf{X}_0^{\top} \eqbf{y}_0\right],
\end{gathered}
\end{equation}
where $\eqbf{X}_0 \in \mathbb{R}^{n_0 \times d}$ and $\eqbf{y}_0 \in \mathbb{R}^{n_0}$ represent the user's existing interaction records.

The motivation behind this method is that the preferences of new users tend to be similar to those of the general public. 
The naive method is to recommend items to users based on the number of positive reviews or popularity of items. While this method may lack personalization, it tends to provide satisfactory results in the early stages of recommendation.
The use of a meta-distribution for user initialization adopts this recommendation strategy. After pretraining, items with a high number of positive reviews will have higher dot product scores with ${u_\text{new}}$, making these well-reviewed items more likely to be recommended. Then, we use the meta-distribution as a prior distribution and adjust it based on the observed online feedback data from users, thus realizing personalized recommendations.

\subsection{Recommend Strategy}

We describe how the system makes a recommendation to the user in each round based on the posterior distribution.    
In the early stages of recommendation, the model's predictions are not sufficiently accurate due to a lack of user interaction data. 
It is necessary to incorporate additional exploration mechanisms that take into account the inherent uncertainty of the model's predictions. 
In this context, we apply the principle of optimism, embodied by the well-known Upper Confidence Bound (UCB), which guides recommendations based on the upper bound of the confidence interval of the predicted score.

At round $t$, the system is given that 
the feature vector of user $u$ 
follows the posterior distribution 
$\mathcal{N}(\eqbf{\mu}_t, \eqbf{\Sigma}_t)$. 
The distribution $\mathcal{N}(\eqbf{\mu}_t, \eqbf{\Sigma}_t)$ quantifies the 
uncertainty in the feature vector of 
user $u$. 
As a consequence, it leads to uncertainty in 
predicting the score of user $u$ toward 
item $i$, which can be quantified as: 
\[
\hat{r}^{(t)}_{u,i} 
= 
\eqbf{e}^{\top}_{u,t}
\eqbf{e}_{i}^{\ast },
\]
where $\hat{r}^{(t)}_{u,i}$ denotes the predicted 
score and 
$\eqbf{e}_{u,t}$ denotes a random 
feature vector with distribution 
$\mathcal{N}(\eqbf{\mu}_t, \eqbf{\Sigma}_t)$.   
We employ information-theoretic to derive confidence bounds on $\hat{r}^{(t)}_{u,i}$. 
The analysis technology for the following result is from \cite{lu2019information}.
\begin{thm}
    \label{thm:ucb}
    At round t, $\forall$ item $i$, with a probability of at least $ 1 - \delta $, the following inequality holds:
    \begin{equation}
        \label{eq:ucb}
        \left|\hat{r}^{(t)}_{u,i} - \mathbb{E} \left[\hat{r}^{(t)}_{u,i}\right]\right| \leq \frac{\Gamma_t}{2} \sqrt{I(u;i, r_{u,i})},
    \end{equation}
    where 
$$
\Gamma_{t}=4 \sqrt{\frac{\lambda_t}{\log \left(1+\frac{\lambda_t}{\sigma_{noise}^2}\right)} \log \frac{2 |\mathcal{A}_t|}{\delta}}, \quad \lambda_t =\max _{i \in \mathcal{A}_t} \eqbf{e}_i^{\ast \top}
\eqbf{\Sigma}_t \eqbf{e}^{\ast}_i,
$$
and $I(u;i, r_{u,i})$ is the ﬁltered mutual information between $u$ and the item-reward pair during the $t$-th round.
Moreover, 
$$
I(u;i, r_{u,i}) = \frac{1}{2} \log\left(1 + \frac{ 
\eqbf{e}_i^{\ast \top}
\eqbf{\Sigma}_t \eqbf{e}^{\ast}_i
}{\sigma_{noise}^2} \right).
$$
\end{thm}

Based on Theorem \ref{thm:ucb}, 
one can derive the UCB of $\hat{r}_{u,i}$ as:
\begin{equation}
\operatorname{UCB}_\delta\left(\hat{r}^{(t)}_{u,i}\right) 
= 
\eqbf{\mu}^\top_t 
\eqbf{e}^\ast_i 
+ \frac{\Gamma_t}{2}\sqrt{ \frac{1}{2} 
\log\left(1 + \frac{\eqbf{e}_i^{\ast \top} \eqbf{\Sigma}_t \eqbf{e}^{\ast}_i}{\sigma_{noise}^2} \right)}.
\label{eq:UCB}  
\end{equation}
Equation \eqref{eq:UCB} can be reduced to the classical LinUCB \cite{li2010contextual}, by combining with the variance $\mathbb{V}\left(\hat{r}^{(t)}_{u,i}\right) = \eqbf{e}^{\ast \top}_i
\eqbf{\Sigma}_t \eqbf{e}^{\ast}_i $, formally 
\begin{equation}
\operatorname{UCB}_\delta\left(\hat{r}^{(t)}_{u,i}\right) \leq \mathbb{E}\left[ \hat{r}^{(t)}_{u,i} \right] + \nu_t \cdot \sqrt{\mathbb{V}\left(\hat{r}^{(t)}_{u,i}\right)},
\label{eq:UCB2}    
\end{equation}
where $\nu_t = \frac{\Gamma_t}{2} \sqrt{\frac{1}{2\noise^2}}$, and $\nu_t$ can be treated as a hyper-parameter in practice \cite{zhao2013interactive, zou2020neural}.

\begin{algorithm}[t]
	\KwIn{User $\eqbf{u}$, parameter $\gamma$,
 item embedding vector $\left\{\eqbf{{e}_{i}^*}:{i\in \mathcal{I}}\right\},$ user's interaction history $\{\eqbf{X}_0, \eqbf{y}_0\}$, $\delta, \sigma_{noise}, T$}
    $\eqbf{\mu_0}, \eqbf{\Sigma_0} \gets $ initialize the user's prior distribution by equation \eqref{eq:init:pre}\;
	\For{t $\gets 1$ \KwTo T}{
	    $\mathcal{A}_t \gets$ construct the set of candidate items to be recommended \; 
	    $\operatorname{UCB}_\delta\left(\hat{r}^{(t)}_{u,i}\right) \gets$ calculate the UCB of the predicted score on $\mathcal{A}_t$ by equation \eqref{eq:UCB} \;
     $i_t \in \operatorname{argmax}_{i \in \mathcal{A}_t} \operatorname{UCB}_\delta\left(\hat{r}^{(t)}_{u,i}\right) $
     \;
	    recommend item $i_t$ to user $u$ and observe user's feedback $r_t$\;
        
            $\eqbf{\mu}_t, \eqbf{\Sigma}_t \gets$ update posterior distribution by equation \eqref{eq:update_dis} \;
        }
	\caption{Online Process of \model }\label{alg:online}
\end{algorithm}

We summarize the process of online recommendation. 
As described in Algorithm \ref{alg:online}, when a user arrives, we initialize the prior distribution based on the meta distribution and the user's interaction history. 
Then, in each round, we maintain a candidate set of recommended items for the user, calculate the Upper Confidence Bound for each item in the set, and make recommendations to the user based on UCB. 
We then collect user feedback, update the user's posterior distribution, and proceed to the next round.
    
\subsection{Regret Analysis}
\label{sec:thm}

Here, we explore the theoretical guarantees of the performance of the proposed methods in online recommendation scenarios. 
Specifically, we focus on the frequently studied regret in the field of bandit algorithms.  
An abundance of literature has investigated the linear bandit and its variants under the Bayesian framework \cite{lu2019information, basu2021no}. 
However, the majority of these studies make their assumptions on the background of known correct prior distributions, an expectation that is unattainable in real-world situations, especially in recommendation scenarios. 
A recent study \cite{peleg2022metalearning} assimilated existing conclusions about Bayesian bandits by investigating the model's regret performance under erroneous or inaccurate priors. 
Our analysis builds on 
the results of \cite{bastani2022meta, peleg2022metalearning} for meta-learning linear bandits. 
  
\noindent
{\bf Setting of regret analysis.}  
This section we make the following assumption about reward generation:
for user $u$, $\forall t \in[T]$, the system recommends 
item $i_t \in \mathcal{I}_{u,t} \subseteq \mathcal{I}$ to user $u$ 
and receives a reward
$$
r_{u,i_t} = \eqbf{e}^\top_u  
\eqbf{e}^\ast_{i_t} + \xi,
$$
where $\xi \sim \mathcal{N}\left(0, \sigma_{noise}^2\right)$.  We consider the general case that 
$\eqbf{e}_u$ follows the Gaussian distribution 
$ \mathcal{N}\left(\eqbf{\mu}_*, \boldsymbol{\Sigma}_*\right)$ with unknown $\eqbf{\mu}_*, \boldsymbol{\Sigma}_*$.  
$\mathcal{I}_{u,t}$ represents the set of item candidates for user $u$ at time $t$, $|\mathcal{I}_{u,t}| \leq N$. 
Typically, it excludes items that have been previously recommended to the user. 
 
For each fixed feature vector of user $u$, i.e., $\eqbf{e}_u$, 
the regret of making $T$ recommendations is defined as: 
$$ 
\text{Reg} 
\left(\eqbf{\mu}_0, \eqbf{\Sigma}_0, T; \eqbf{e}_u \right) 
\triangleq \sum_{t=1}^T \mathbb{E}\left[r_{u, i_{t}^*} - r_{u, i_{t}} \right],
$$
where $i_{ t}^\ast 
=
\operatorname{argmax}_{i \in \mathcal{I}_{u,t}} 
\eqbf{e}^\top_u \eqbf{e}_i. $
Then, the Bayesian regret is defined as 
\[
\text{Reg}_{\text{Bay}} 
\left(\eqbf{\mu}_0, \eqbf{\Sigma}_0, T \right)
\triangleq 
{\mathbb{E}_{\eqbf{e}_u \sim 
\mathcal{N} (\eqbf{\mu}_*, \eqbf{\Sigma}_*) } }
[
\text{Reg}(\eqbf{\mu}_0, \eqbf{\Sigma}_0, T; \eqbf{e}_u)
].
\]
 
\begin{assumption}
\label{ass:bound}
We make the following common boundedness assumption for both the item feature vectors and the user's prior distribution:
\begin{enumerate}[label=(\alph*)]
  \item 
  For any item $i, \|\eqbf{{e}_i}^*\|_2 \leq a$.  
  The embedding vectors 
  $\eqbf{{e}_i}^\ast, \forall i$, are i.i.d. samples 
  from a truncated zero mean Gaussian distribution with 
  covariance matrix $\boldsymbol{\Sigma}_{{A}}$ and 
  support $\{\eqbf{x} : \eqbf{x} \in \mathbb{R}^d, 
  \|\eqbf{x} \|_2 \leq a 
  \}$.  
  Furthermore, the covariance matrix $\boldsymbol{\Sigma}_{{A}}$ satisfies 
  $\lambda_{\min }\left(\boldsymbol{\Sigma}_{{A}}\right) \geq \lambda_{\boldsymbol{\Sigma}_{{A}}}>0$. 
  \item For user's prior mean $\|\eqbf{\mu}_*\|_2 \leq m$, 
and the minimal and maximal eigenvalues of the prior covariance matrix are lower and upper bounded by known constants
\[
0 < \underline{\lambda}  \leq  \lambda_{\min }\left(\boldsymbol{\Sigma}_*\right) 
\leq 
\lambda_{\max }\left(\boldsymbol{\Sigma}_* \right) \leq \bar{\lambda}.
\]
\end{enumerate}
\end{assumption}

\begin{definition}
\label{def:tau}
    We introduce sufficient rounds $\tau$ defined as,
    $$
    \tau = \min\left\{t:\lambda_{min}(\eqbf{V_t}) \geq \frac{\lambda_{\eqbf{\Sigma_A}}d}{2}\right\},
    $$
where $\eqbf{V_t} = \sum_{s=1}^t {\eqbf{e}^*_{i_s}} {\eqbf{e}^{*\top}_{i_s}}$.
\end{definition}

\begin{remark}
    $\tau$ is the number of rounds required for the algorithm to sufficiently explore in all directions. 
    The concept of $\tau$ has appeared in previous work \cite{peleg2022metalearning}. 
    However, to ensure sufficient exploration in all directions, the previous approach involved multiple rounds of completely random exploration in the early stages. 
    This is not feasible in real-world recommendation scenarios, as it is tantamount to sacrificing the user's initial satisfaction. 
    Therefore, we introduce a new concept for sufficient rounds in this context.
\end{remark}

\noindent
{\bf Regret upper bound.} 
We follow common practice in the bandit literature of dividing random events into the set of ``good events'' and their complement \cite{lattimore2020bandit}. 
Similar with previous work \cite{peleg2022metalearning}, 
for any $\delta > 0$, we defined the good event $\mathcal{E}$ as 
$\mathcal{E} 
\triangleq 
\left\{\mathcal{E}_1 \cap \mathcal{E}_2 \cap \mathcal{E}_3
\right\}$, 
where events $\mathcal{E}_1, \mathcal{E}_2$ 
and $\mathcal{E}_3$ are defined as: 
\[
\begin{aligned}
& 
\mathcal{E}_1 
\triangleq 
\left\{
\left\|
\eqbf{\mu}_0 -\eqbf{\mu}_{*} \right\| 
\leq 
\sqrt{c_1 \delta}\right\}, \\
& 
\mathcal{E}_2 
\triangleq 
\left\{
\left\|
\boldsymbol{\Sigma}_0
-\boldsymbol{\Sigma}_* 
\right\|_{\mathrm{op}} \leq \sqrt{c_2 \delta}, \quad \boldsymbol{\Sigma}_0 \succeq \boldsymbol{\Sigma}_* \right\}, 
\\
& 
\mathcal{E}_3 
\triangleq
\left\{ 
\left\| 
\boldsymbol{\Sigma}_u^{-1 / 2} 
\left(\eqbf{e}_u -\eqbf{\mu}_*\right) 
\right\|_{\infty}^2 
\leq 
2 \ln \left(\frac{d^2 T}{\delta}\right)\right\}, 
\end{aligned}
\]
where $c_1>0$ and $c_2 >0$ are 
hyper-parameters to be selected later. 
The events $\mathcal{E}_1, \mathcal{E}_2$ represent the distance between the prior of algorithm 
$\mathcal{N}(\eqbf{\mu}_0, \boldsymbol{\Sigma}_0)$ and the true unknown prior $\mathcal{N}({\eqbf{\mu}_*}, {\boldsymbol{\Sigma}_*}).$  
The event $\mathcal{E}_3$ is an instance-based event, 
unrelated to the performed algorithm.  
It represents the event at the realization of $\eqbf{e}_u$ is not too far from its mean. 

To facilitate the regret analysis, 
one needs to first derive sufficient conditions 
under which the good event $\mathcal{E}$ 
holds with high probability.  
This involves not only appropriate selections of 
the hyper-parameters $c_1$ and $c_2$, 
but also the selection of the hyper-parameter $\gamma$ 
which controls the width of $\eqbf{\Sigma}_0$, 
i.e., the covariance of prior distribution of $\eqbf{e}_u$.  
Some arguments in \cite{bastani2022meta} assist 
the selection of $\gamma$.  
Selecting $c_1$ and $c_2$ is technically non-trivial, and we leave the details of the selection process in 
the Appendix.  
We summarize appropriate selections of them in 
the following lemma.  
\begin{lemma}
\label{lemma:event}
The event $\mathcal{E}$ holds with probability larger than $1-\frac{9 \delta}{d T}$, for $M \geq 5d + 2 \ln \left(\frac{dMT}{3}\right)$, by setting $\delta = 1/M$, 
\begin{equation}
\label{eq:choose_f}
\begin{aligned}
& 
\gamma= 
32 \Bar{\lambda} \cdot \sqrt{\frac{5d + 2 \ln \left(\frac{dMT}{3}\right)}{M}}, 
\\
& c_1 =  {{\Bar{\lambda}}\left(2d + 3\ln{{\left(d M T\right)}} \right)}, \\
& c_2 =  (64 \Bar{\lambda})^2 \left({{5d + 2 \ln \left(\frac{dMT}{3}\right)}}\right). \\
\end{aligned}
\end{equation} 
\end{lemma}

Based on Lemma \ref{lemma:event}, we prove an upper bound 
on the Bayesian regret, and we reserve the proof to the 
Appendix of this paper.  
 
\begin{thm}
\label{thm:regret}
Suppose $\gamma, c_1$ and $c_2$ satisfy Equation 
(\ref{eq:choose_f}).  
The Bayesian regret of Algorithm \ref{alg:online} is bound by: 
$$
\begin{aligned}
\text{Reg}_{\text{Bay}} 
\left(\eqbf{\mu}_0, \eqbf{\Sigma}_0, T \right)
\leq (1+k_1)\left( \Gamma \sqrt{ \frac{1}{2}Td \log \left(1+\frac{\Bar{\lambda}T}{\sigma_{noise}^2}\right)} + B\right)  + \frac{c_{bad}\delta}{\sqrt{d}} +k_2 \tau,
\end{aligned}
$$
where $ \Gamma = 4 \sqrt{\frac{\Bar{\lambda}}{\log \left(1+\frac{\Bar{\lambda}}{\sigma_{noise}^2}\right)} \log (4 N T)}$, 
$c_{b a d} = 22 a\left(m+\sqrt{4 \bar{\lambda} \ln \left({d^2MT}\right)}\right)$,
$ B = a \left(m + \sqrt{\Bar{\lambda} d } \right) $, $ k_2 = 2B$, $\delta = 1 / M$, and $k_1 \in \tilde{\mathcal{O}}\left(\sqrt{c_1 \delta}+\tau \sqrt{c_2 \delta}\right)$ defined in lemma \ref{lemma:meta}.
\end{thm}
The regret bound in the above Theorem \ref{thm:regret} depends on the total number of rounds T, the sufficient round $\tau$, and some constants.
The order with respect to $T$ is $\tilde{\mathcal{O}}(\sqrt{T})$, which is a sublinear result, meaning that as the number of rounds increases, the average regret tends to zero. 
The factor that affects $\tau$ is the exploration strategy used. The principle of optimism in the face of uncertainty tends to choose directions that have not been sufficiently explored. Intuitively, using a UCB strategy can help us achieve the goal of sufficient exploration with fewer rounds.
Through a more detailed analysis of $\tau$, there may be further improvement for algorithm, which will be left for future work.

\section{Experiments}
\label{sec:exp}

In this section, we conduct extensive experiments on three datasets to evaluate the effectiveness of \model.
Particularly, our experiments aim to answer the following research questions:
\begin{itemize}[leftmargin=*]
    \item \textbf{Q1}: How can \model \ outperform existing interactive collaborative filtering algorithms for the cold-start users?
    \item \textbf{Q2}: Can the \model \ be applied to warm-start users with drifting taste, i.e., those whose interests change over time?
    \item \textbf{Q3}: Considering top-k recommendation over time, can the algorithms still be effective?
    \item \textbf{Q4}: What’s the influence of various components in \model?
    \item \textbf{Q5}: How do the key hyperparameter settings impact \model’s performance?
\end{itemize}
In the following subsections, we first present the experimental settings and then answer the above research questions in turn.

\subsection{Experimental Settings}

\noindent 
{\bf Datasets.}   
We evaluate the proposed method on three real-world datasets, namely  KuaiRec, Movielens(1M), and EachMovie. The statistical information of the datasets is summarized in Table \ref{tab:stat}. 
These datasets have also been widely used in related studies \cite{zou2020neural, zhao2013interactive, gao2022kuairec}.

MovieLens and EachMovie are movie rating datasets that are widely used for performance comparison of recommendation algorithms. 
The interaction records in these datasets consist of integer ratings ranging from 0 to 5. 
Following previous work \cite{zou2020neural, zhao2013interactive}, we consider ratings greater than or equal to 4 as satisfied interactions.
KuaiRec \cite{gao2022kuairec} is a short video dataset where the interaction records represent the duration of user video views. 
Following the suggestion of the authors \cite{gao2022kuairec}, we consider video views with a duration greater than twice the length of the video as satisfied interactions.

It is worth noting that KuaiRec is a fully-observed dataset, meaning that we have access to complete information about user-item interactions. 
This allows us to address the challenge of offline evaluation effectively. 
For any recommendations made by the model to users in this dataset, we can always query the actual feedback information. 
However, for the other two datasets (MovieLens and EachMovie), users are unlikely to interact with all items. 
Therefore, in offline scenarios, we cannot obtain real feedback information for the model's recommendations if there is no corresponding interaction record in the dataset. 
To handle this, we follow previous work \cite{zou2020neural, zhao2013interactive} and fill the missing interaction records with 0, indicating no interaction.

\begin{table}[ht]
\small
\centering
\setlength{\tabcolsep}{3pt}
\caption{The Statistics of Datasets.}\label{tab:stat}
\begin{tabular}{l|c|c|c}
\hline Dataset & KuaiRec & MovieLens (1M) & EachMovie \\
\hline
\hline
\# Users & $1,411$ & $6,040$ & $61,265$ \\
\# Items & $3,327$ & $3,706$ & $1,623$ \\
\# Interactions & $4,676,570$ & $1,000,209$ & $2,811,718$ \\
\# Interactions Per User & $3314.37$ & $165.60$ & $1732.42$ \\
\# Interactions Per Item & $1405.64$ & $269.89$ & $45.89$ \\
\hline
\end{tabular}
\end{table}

\noindent
{\bf Baselines.} 
In this part, we introduce the baseline methods for comparison. The compared methods are as follows.
\begin{itemize}
    \item \textbf{Random}: In each interaction, randomly chooses an item from the entire item set to recommend to the target user. It is a baseline used to estimate the worst performance that should be obtained
    \item \textbf{Pop}: The system picks the most popular items to recommend to the target user. This is a commonly employed basic baseline. Despite lacking personalization, it performs surprisingly well in evaluations, as users tend to consume popular items.
    \item \textbf{ICF\cite{zhao2013interactive}}: Interactive collaborative filtering combines probabilistic matrix factorization\cite{mnih2007probabilistic} with various exploration methods for recommender system, including LinUCB\cite{li2010contextual}, 
    and TS\cite{chapelle2011empirical}
    \item \textbf{MF\cite{koren2009matrix}}: We always greedy w.r.t. the estimated scores and update users’ latent factor after every interaction. It is regard as the myopic algorithm of ICF.
    \item \textbf{NICF\cite{zou2020neural}}: A deep reinforcement learning method used to address interactive collaborative filtering. We use the implementation provided by the authors\footnote{https://github.com/zoulixin93/NICF}.
    \item \textbf{\model}: Our proposed method.
\end{itemize}

\noindent 
{\bf Evaluation metrics.} 
Consistent with previous work \cite{zhao2013interactive}, three evaluation metrics are used:
\begin{itemize}
    \item \textbf{Cumulative Precision}@T. A straightforward measure is the number of positive interactions collected during the total $T$ interactions,
    $$
    \text { precision@T }=\frac{1}{\# \text {users}} \sum_{\text {users}} \sum_{t=1}^T \theta_t.
    $$
    For ratings datasets (MoveieLens, EachMovie), we define $\theta_t = 1$ if $r_{u, i_t}>=4$, and 0 otherwise.
    For video datasets (KuaiRec), we set $\theta_t = r_t$.
    
    \item \textbf{Cumulative Recall}@T. We can also check for the recall during $T$ timesteps of the interactions,
    $$
    \text {recall@T} = \frac{1}{\# \text {users}} \sum_{\text {users}} \sum_{t=1}^T \frac{\theta_t}{\text{ \# satisfied items }}.
    $$

    \item \textbf{Cumulative $\eqbf{nDCG_k}$}@T. For the case that multiple items are shown in one interaction, the ranking of the item listed is also important: it is more useful to have the highly relevant items appear earlier in the ranking list. We use the normalized discounted cumulative gain ($nDCG_k$) as the ranking measure
    $$
    nDCG_{k} = \frac{1}{\mathcal{Z}} \sum_{j = 1}^k \frac{2^{\theta_{t, j}} - 1}{\log_2(1 + j)},
    $$
    where $\theta_{t, j} $ is the real feedback $\theta_t$ of the item shown at ranking position $j$ in round $t$. $\mathcal{Z}$ is the normalization factor making the score of the optimal ranking to 1 such that $0 \leq nDCG_k \leq 1$. 
    Similar to the cumulative precision and recall, here the cumulative $nDCG_k$ should also take sum over $T$ and average on users,
    \[
    nDCG_{k}@T = \frac{1}{\# \text {users}} \sum_{\text {users}} \sum_{t=1}^T nDCG_{k}.
    \]
\end{itemize}

\noindent
{\bf Parameter setting.}
For all datasets, we use 50\% of the interaction data as the training set. Test users are selected outside of the training set, according to the different goals of each experiment.
For all methods except Random and Pop, a grid search is used to find the optimal configurations. 
For ICF, MF, and the proposed method, the latent dimensions \(d\) are chosen from the set \{32, 64, 128, 256\}. The maximum number of alternation optimization rounds for ICF is set to 20. 
In the proposed method, the depth \(K\) of the graph neural network is fixed at 3, \(\gamma\) is chosen from the set \{0.01, 0.1, 1 \}, and the learning rate is picked from the set \{0.01, 0.1, 0.5, 1, 5, 10\}.
We consider \(v_t\) in equation \eqref{eq:UCB2} as a tunable hyper-parameter, selecting its value from the set \{0.1, 0.5, 1, 5, 10\}.
As for the NICF, we use the same configuration as described in their paper \cite{zou2020neural}, where the initial dimension is set to 50, two attention blocks are used, and the optimal experimental results are reported from a selection of 3000 epochs.
We report the result of each method with its optimal hyper-parameter settings.

\subsection{Performance Comparison on cold-start cases (Q1)}
In this experiment, to evaluate the performance of the algorithm on cold start users, we selected 200 users with the highest number of interactions as test users. 
Their interaction data was excluded from the training data, ensuring that their previous interactions were not seen during the training process. 
We looked specifically at how well the different methods performed in recommending items to these users over a period of 120 interactions.

The experimental results are presented in Table \ref{tab:cold-s}. 
We ran our method and the comparison methods ten times and report the best results. The best-performing method is highlighted in bold and marked with an asterisk $*$ to indicate significant improvement over the best baseline, as determined by the Wilcoxon signed-rank test with the p-value less than $0.05$.
To summarize, our findings are as follows:
In the context of cold-start problems, the proposed method outperforms existing techniques in terms of recall and precision. In particular, on average, the relative improvement in cumulative $precision@120$ over the best baseline is 3.34\%, 3.84\%, and 10.66\%, respectively, for the three benchmark datasets.
Among the existing methods, ICF-TS does not perform as effectively as UCB methods in the early stages of cold start scenarios. However, its long-term performance is similar to that of UCB-type methods. NICF shows strong results initially, but its performance declines over a longer period of time.

\begin{table}[t]
\centering
\small
\caption{Cold-start recommendation performance of different models.}
\label{tab:cold-s}
\resizebox{\textwidth}{!}{
\begin{tabular}{c|cccc|cccc|cccc}
\hline 
Dataset & \multicolumn{4}{|c|}{ KuaiRec } & \multicolumn{4}{c|}{ MovieLens (1M) } & \multicolumn{4}{c}{ EachMovie } \\
\hline 
\hline Measure & \multicolumn{4}{|c|}{ Cumulative Precision } & \multicolumn{4}{c|}{ Cumulative Precision } & \multicolumn{4}{c}{ Cumulative Precision } \\
\hline 
T & 10 & 20 & 40 & 120 & 10 & 20 & 40 & 120  & 10 & 20 & 40 & 120  \\
\hline 
\hline 
Random & 0.435 & 0.925 & 2.055 & 5.870 & 0.900 & 1.830 & 3.535 & 10.98 & 1.730 & 3.350 & 6.750 & 20.53 \\
Pop & 0.935 & 1.780 & 3.400 & 8.870 & 7.200 & 13.34 & 25.78 & 64.16 & 6.245 & 12.33 & 23.55 & 58.58 \\
MF & 1.810 & 4.055 & 7.425 & 15.17 &3.730 &7.900 &17.68 &56.24 &3.125 &6.715 &16.74 &50.96\\
ICF-UCB & 6.240 & \underline{10.55} & \underline{14.88} & 25.99 & 7.410 & 14.48 & \underline{27.06} & \underline{66.04} & 6.945 & 12.72 & 24.81 & \underline{63.32}\\
ICF-TS & 4.085 & 7.885 & 13.98 & \underline{26.38} & 5.180 & 10.84 & 22.27 & 62.00 & 5.145 & 10.46 & 20.82 & 57.43 \\
NICF & \underline{6.365} & 10.49 & 14.69 & 22.68 & \underline{7.505} & \underline{15.01} & 26.18 & 57.63 & \underline{7.260} &\underline{13.62}	&\underline{25.10}	& 51.17 \\
\model & $\mathbf{6.405}^*$ & $\mathbf{10.61}^*$ & $\mathbf{15.18}^*$ & $\mathbf{27.26}^*$ & $\mathbf{7.535}^*$ & $\mathbf{15.36}^*$ & $\mathbf{27.31}^*$ & $\mathbf{68.58}^*$ & $\mathbf{7.550}^*$ & $\mathbf{14.34}^*$ & $\mathbf{27.27}^*$ & $\mathbf{70.07}^*$\\
\hline 
\hline 
Measure & \multicolumn{4}{|c|}{ Cumulative Recall } &  \multicolumn{4}{|c}{Cumulative Recall} &  \multicolumn{4}{|c}{Cumulative Recall} \\
\hline 
T & 10 & 20 & 40 & 120 & 10 & 20 & 40 & 120  & 10 & 20 & 40 & 120  \\
\hline 
\hline
Random & 0.0025 & 0.0053 & 0.0124 & 0.0342 & 0.0027 & 0.0053 & 0.0102 & 0.0311 & 0.0063 & 0.0120 & 0.0242 & 0.0741 \\
Pop & 0.0051 & 0.0102 & 0.0199 & 0.0497 & 0.0246 & 0.0446 & 0.0855 & 0.2044 & 0.0244 & 0.0481 & 0.0901 & 0.2230 \\
MF & 0.0141 & 0.0311 & 0.0526 & 0.0936 & 0.0117 & 0.0251 & 0.0553 & 0.1790 &0.0113 & 0.0257 & 0.0564 & 0.1730 \\
ICF-UCB & 0.1021 & \underline{0.1547} & \underline{0.1881} & 0.2551 & 0.0251 & 0.0482 & \underline{0.0909} & \underline{0.2133} & 0.0270 & 0.0492 & 0.0961 & \underline{0.2422} \\
ICF-TS & 0.0524 & 0.1016 & 0.1744 & \underline{0.2580} & 0.0165 & 0.0350 & 0.0729 & 0.1995 & 0.0202 & 0.0404 & 0.0794 & 0.2180 \\
NICF & \underline{0.1040} & 0.1540 & 0.1847 & 0.2320 & \underline{0.0254} & \underline{0.0489} & 0.0869 & 0.1831 & \underline{0.0284} &\underline{0.0530} &\underline{0.0967} &0.1947\\
\model & $\mathbf{0.1052}^*$ & $\mathbf{0.1554}^*$ & $\mathbf{0.1909}^*$ & $\mathbf{0.2652}^*$ & $\mathbf{0.0261}^*$ & $\mathbf{0.0496}^*$ & $\mathbf{0.0918}^*$ & $\mathbf{0.2204}^*$ & $\mathbf{0.0300}^*$ & $\mathbf{0.0571}^*$ & $\mathbf{0.1073}^*$ & $\mathbf{0.2716}^*$\\
\hline
\end{tabular}
}
\end{table}

\subsection{Performance Comparison on warm-start cases with taste drift (Q2)}
In this experiment, our goal is to investigate whether the algorithms can effectively adapt to warm-start users and track their changing interests over time. 
For each user, we divide their rating records into two equal-sized periods, referred to as set 1 and set 2. The interactions in set 1 occurred earlier in time compared to set 2.
Following previous work \cite{zou2020neural,zhao2013interactive, shi2012adaptive}, to capture the users' interest drift, we utilize the genre information of the items as an indication. 
Specifically, we calculate the cosine similarity between the genre/categories vectors of the two periods. 
Users with the smallest cosine similarity are considered to exhibit significant interest drift between the two time periods. 
The remaining users, along with their ratings, form the training set. 
We conduct experiments on the KuaiRec and MovieLens datasets, as the EachMovie dataset does not provide genre information for movies.
For each test user, in the ﬁrst period with 60 interactions, we use set 1 as the ground truth of the test users; and then, from the 61st interaction, the ground truth is changed from set 1 to set 2 to simulate the process of his/her taste drift.

The experimental results are presented in Table \ref{tab:warm-s}. 
As we focus on the performance when the user has changed the interest, only the results for $T \ge 60$ are shown. All other settings are the same as for cold start experiments.
To summarize, our findings are as follows:
Our proposed methods show superior performance relative to the baselines on both datasets. The improvement over the best performing baseline reaches as high as 4.31\% for the KuaiRec dataset and 5.79\% for the MovieLens (1M) dataset. This suggests that for warm-start users, our proposed methodology is able to monitor changes in user preferences and adjust its strategy to better meet user needs.

\begin{table}[t]
\centering
\small
\caption{Performance on Warm-start Users with Taste Drift on KuaiRec and MovieLens.}
\label{tab:warm-s}
{
\begin{tabular}{c|cccc|cccc|cccc}
\hline 
Dataset & \multicolumn{4}{|c|}{ KuaiRec } & \multicolumn{4}{c}{ MovieLens (1M) } \\
\hline 
\hline Measure & \multicolumn{4}{|c|}{ Cumulative Precision } & \multicolumn{4}{c}{ Cumulative Precision } \\
\hline 
T & 60 & 80 & 100 & 120 & 60 & 80 & 100 & 120 \\
\hline 
\hline 
Random & 2.02 & 2.71 & 3.27 & 3.90 & 1.39 & 1.92 & 2.40 & 2.77 \\
Pop & 5.62 & 5.87  & 5.88 & 5.95 & 15.43 & 17.89 & 20.24 & 22.53 \\
MF & 3.15 & 3.63 & 4.51 & 4.91 &2.93 &3.71 &5.67 &7.29\\
ICF-UCB & 13.84 & {14.59} & {15.15} & 15.71 & \underline{15.63} & \underline{18.05} & {20.44} & {23.18} \\
ICF-TS & \underline{13.87} & \underline{14.63} & \underline{15.18} & \underline{15.77} & 14.51 & 17.66 & \underline{20.46} & \underline{23.66} \\
NICF & {6.36} & 7.41 & 7.93 & 8.05 & {7.96} & {13.03} & 16.05 & 18.58 \\
\model & $\mathbf{14.87}^*$ & $\mathbf{15.56}^*$ & $\mathbf{15.83}^*$ & $\mathbf{16.45}^*$ & $\mathbf{16.55}^*$ & $\mathbf{19.26}^*$ & $\mathbf{22.25}^*$ & $\mathbf{25.03}^*$\\
\hline 
\hline 
Measure & \multicolumn{4}{|c|}{ Cumulative Recall } &  \multicolumn{4}{|c}{Cumulative Recall} \\
\hline 
T & 60 & 80 & 100 & 120 & 60 & 80 & 100 & 120 \\
\hline 
\hline
Random & 0.0083 & 0.0110 & 0.0135 & 0.0166 & 0.0078 & 0.0110 & 0.0136 & 0.0159 \\
Pop &  0.0257 & 0.0268 & 0.0269 & 0.0271  & 0.0994 & 0.1146 & 0.1302 & 0.1439 \\
MF & 0.0149 & 0.0179 & 0.0213 & 0.0243 &0.0151 &0.0164 &0.0283 &0.0533\\
ICF-UCB & 0.1570 & {0.1604} & {0.1629} & 0.1653 & \underline{0.1023} & \underline{0.1153} & {0.1310} & {0.1464} \\
ICF-TS & \underline{0.1573} & \underline{0.1616} & \underline{0.1638} & \underline{0.1662} & 0.0955 & 0.1140 & \underline{0.1315} & \underline{0.1470} \\
NICF & {0.0591} & 0.0632 & 0.0648 & 0.0651 & 0.0619 & 0.0811 & 0.1036 & 0.1190 \\
\model & $\mathbf{0.1642}^*$ & $\mathbf{0.1672}^*$ & $\mathbf{0.1683}^*$ & $\mathbf{0.1712}^*$ & $\mathbf{0.1048}^*$ & $\mathbf{0.1228}^*$ & $\mathbf{0.1421}^*$ & $\mathbf{0.1599}^*$\\
\hline
\end{tabular}
}
\end{table}

\subsection{Top-K Ranking Performance (Q3)}
In this part, we conduct experiments with multiple item slots at each interaction. 
The common ranking-aware measure nDCG is used to test the performance. 
The test users are the same as the ones in the cold-start setting. 
The only diﬀerence is that the number of interactions is reduced since the number of recommended items in each interaction increases. 
Since NICF is a reinforcement learning method specifically designed for single-item recommendations, we have not included it in our comparison here.

The experimental results are presented in Table \ref{tab:topk}. 
We adjusted the number of recommended items in each round of the cold start experiments to either $3$ or $5$, while keeping all other settings the same.
To summarize, our findings are as follows:
A similar trend is shown compared to the case of one item, 
the proposed method outperforms existing techniques in terms of nDCG,
the relative improvement in cumulative $nDCG_3@40$ over the best baseline is 2.67\%, 4.42\%, and 10.66\%, respectively, for the three benchmark datasets.

\begin{table}[ht]
    \centering
    \caption{Performance on Top-K Recommendations by Cumulative nDCG.}
    \resizebox{\textwidth}{!}{
    \begin{tabular}{c|cc|cc|cc|cc|cc|cc}
    \hline
    Dataset & \multicolumn{4}{|c|}{ KuaiRec } & \multicolumn{4}{c|}{ MovieLens (1M) } & \multicolumn{4}{c}{ EachMovie } \\
    \hline
    Measure & \multicolumn{2}{|c|}{ $nDCG_3@T$ } & \multicolumn{2}{c}{$nDCG_5@T$} & \multicolumn{2}{|c|}{ $nDCG_3@T$ } & \multicolumn{2}{c}{$nDCG_5@T$} & \multicolumn{2}{|c|}{ $nDCG_3@T$ } & \multicolumn{2}{c}{$nDCG_5@T$} \\    
    \hline 
    T & 20 & \multicolumn{1}{c|}{40} & 10 & 20 & 20 & 40 & 10 & 20 & 20 & 40 & 10 & 20 \\
    \hline
    \hline Random & 1.04 & 1.98 & 0.48 & 0.97 & 1.91 & 3.84 & 0.91 & 1.89 & 3.37 & 6.83 & 1.67 & 3.37\\
    Pop & 1.64 & 2.99 & 0.79 & 1.51 & 12.21 & 21.25 & \underline{6.25} & 10.96 & 11.30 & 19.59 & 5.88 & 10.28 \\
    MF & 4.29 & 6.80 & 2.01 & 3.68 & 8.55 & 16.52 & 4.12 & 8.47 & 9.57 & 18.58 & 4.70 & 9.50 \\
    ICF-UCB & \underline{6.14} & 8.57 & \underline{3.41} & {4.76} & \underline{12.30} & \underline{21.68} & \underline{6.25} & \underline{11.22} & \underline{11.93} & \underline{21.19} & \underline{6.08} & \underline{11.04} \\
    ICF-TS & 6.06 & \underline{8.98} & 3.14 & \underline{4.94} & 9.70 & 19.45 & 4.68 & 9.69 & 9.98 & 18.89 & 4.96 & 9.54 \\
    \model & $\mathbf{6.38}^*$ & $\mathbf{9.22}^*$ & $\mathbf{3.47}^*$ & $\mathbf{5.09}^*$ & $\mathbf{12.50}^*$ & $\mathbf{22.64}^*$ & $\mathbf{6.47}^*$ & $\mathbf{11.85}^*$ & $\mathbf{13.47}^*$ & $\mathbf{23.45}^*$ & $\mathbf{6.86}^*$ & $\mathbf{12.50}^*$\\
    \hline
    \end{tabular}
    }
    \label{tab:topk}
\end{table}

\subsection{Model Ablation and Hyperparameter Sensitivity Studies (Q4, Q5)}
In this subsection, we conduct experiments to investigate the role of different components within our proposed method, as well as the impact of key parameters on the model's performance. 
Compared to traditional techniques, our improvements are manifested in three aspects: 
1) in the pretraining phase, we use graph neural network aggregation operations to link items and users; 
2) in the online phase, we use a meta-distribution as initialization; 
3) we incorporate exploration techniques in the online phase.
The key hyperparameter in our experiments is the depth $K$ of the graph neural network.
The experimental results are presented in Table \ref{tab:ablation}. 

To investigate the improvements provided by the graph aggregation component, we ran experiments under the ``- Aggregation ($K=0$)'' setting, i.e., the depth of the graph neural network was set to 0. This reduces the pretraining method to the PMF \cite{mnih2007probabilistic} optimized by SGD. To investigate the impact of the meta-distribution, we conducted experiments under the ``- Meta" setting, where the initialization of the online user distribution was done with $\mathcal{N}(0, \sigma_0^2)$, as opposed to using a meta-distribution. To examine the contribution of the exploration technique, we ran experiments under the ``- Exploration" setting, where the degree of exploration during the online phase was set to zero. We compared these results with the default settings, which include the full model with all three components. The experimental results show that the performance of the model decreases regardless of which component is omitted.

To evaluate the combined effect of the meta-distribution and exploration techniques, we also ran the ``- Meta \& Exploration" experiment. The results indicate that the absence of additional techniques in the online phase can lead to a significant performance degradation.

To investigate the influence of the key hyper-parameter $K$, we ran ``$K=1$" and ``$K=2$" experiments. In the default method, the $K$ is set to $3$. The results suggest that the performance of the model improves with increasing depth. When the neighborhood information is not aggregated, i.e., when ``$K=0$" is set, the performance of the model deteriorates significantly.

\begin{table}[ht]
    \centering
    \caption{ Model ablation and effects of key hyperparameters on KuaiRec}
    \begin{tabular}{l|ccc}
    \hline Method & $Precsion@60$ & $Recall@60$ & $nDCG_3@20$ \\
    \hline \hline 
    Default & 18.94 & 0.2170 & 6.38 \\
     - Meta & 18.01 & 0.2098 & 6.28 \\
     - Exploration & $18.78$ & 0.2150 & 6.33 \\
     - Meta \& Exploration & $17.52$ & 0.2086 & 6.25 \\
     - Depth ($K$=2) & $18.85 $ & 0.2152 & 6.34 \\
     - Depth ($K$=1) & $18.80 $ & $0.2151 $ & $6.33 $ \\   
     - Aggregation ($K$=0)  & $17.74 $ & $0.2102$ & $ 6.27 $ \\       
    \hline
    \end{tabular}
    \label{tab:ablation}
\end{table}

\section{Related Work}
\label{sec:works}
In this section, we review the related work. 
Our work is mainly related to Collaborative Filtering, Interactive Recommender System and Bayesian Bandit. 

\subsection{ Collaborative Filtering}
In the field of modern recommender systems, collaborative filtering (CF) plays a prominent role \cite{covington2016deep, ying2018graph}. 
CF models such as Matrix Factorization (MF) \cite{koren2009matrix} were commonly employed, which used an embedding vector projection of a user or item ID to encapsulate users and items as embeddings and thus reconstructing historical interactions.
As the field evolved, the advent of neural network-based recommender models such as NCF \cite{he2017neural} and LRML \cite{tay2018latent} brought about a shift. Although these models retained the use of the embedding component, they greatly enhanced the interaction modeling mechanism by exploiting the ability of neural networks to model complex interactions.
More recently, inspired by the power of graph convolution, new methods such as LightGCN \cite{he2020lightgcn}, GC-MC \cite{berg2017graph}, PinSage \cite{ying2018graph}, SiGRec \cite{HUANG2023103403}, and XSimGCL \cite{10158930} have been developed that adapt GCN to the user-item interaction graph for recommendations. 
These graph neural network-based models capture CF signals from high-hop neighbors, illustrating a significant leap forward in the field of recommendation systems.
Our work effectively integrates the LightGCN \cite{he2020lightgcn} model widely used in the field.

\subsection{Interactive Recommender System}
There are two major method, contextual bandit and reinforcement learning, for Interactive recommender system:
The Contextual Bandit approach focuses primarily on the application of bandit technology to various scenarios and developing theoretical results. 
Numerous recommender systems based on the Contextual Bandit have been developed to address different recommendation tasks. 
These include news recommendation \cite{li2010contextual}, collaborative filtering \cite{zhao2013interactive}, and online advertising \cite{du2021exploration, wu2021adversarial, guo2020deep}.
On the other hand, Reinforcement Learning methods focus on developing efficient technologies to overcome the challenges inherent in direct RL applications, such as off-policy training \cite{zou2020pseudo}, off-policy evaluation \cite{gilotte2018offline}, and handling large action spaces \cite{dulac2015deep}. 
The focus of these topics is the optimization of metrics with delayed attributes \cite{zou2019reinforcement, zou2019reinforcement2}.
NICF \cite{zou2020neural}, as a RL method, ingeniously integrates modified self-attention blocks and Q-learning, successfully applying RL to the domain of Interactive Collaborative Filtering. 
Our proposed method belongs to the bandit method within the field of collaborative filtering.

\subsection{Bayesian Bandit}
We review theoretical work conducted under Bayesian settings in recent years.
\cite{kveton2021meta} focus on a fully Bayesian multi-armed bandits (MAB) setting, where tasks are drawn from a Gaussian prior. The prior is parameterized by a known scalar covariance and an unknown mean, that is itself drawn from a known hyper-prior. The authors derive a regret bound which depends on $\tilde{\mathcal{O}}(T^2)$.
\cite{basu2021no} assume a fully Bayesian framework where the covariance is known and the mean is sampled from a known Gaussian distribution.
They establish a prior-dependent regret bound whose worst-case dependence on $T$ is $\tilde{\mathcal{O}}(\sqrt{T})$.
\cite{simchowitz2021bayesian} bound the single instance misspeciﬁcation error for a wide class of priors and settings and achieve an upper-bound of $\tilde{\mathcal{O}}(\varepsilon T^2)$, where $\varepsilon$ is the initial total-variation prior  estimation error.
\cite{peleg2022metalearning} assume the expected rewards originate from a vector, sampled from a Gaussian distribution with unknown mean and covariance. They derive a regret bound that depends on $T$ as $\tilde{\mathcal{O}}(\sqrt{T})$, at the cost of sacrificing the first $\tau$ rounds for random exploration.
Our theoretical results are based on the results of \cite{peleg2022metalearning}, which most closely resemble real-world recommendation scenarios.

\section{Conclusion}
\label{sec:con}
In this paper, we propose a novel method iGCF that extends the ICF and addresses the shortcomings of existing bandit methods, the challenges posed by the cold-start problem and data sparsity. 
Our proposed method combines bandit techniques with state-of-the-art graph neural networks, which effectively enhance the collaborative filtering between users and items. This enhancement significantly improves the expressiveness and performance of the model. 
To overcome the computational hurdles posed by nonlinear models, we incorporate variational inference techniques into the method, ensuring analytical computation even in the complex context of probabilistic models.
In addition, we introduce a meta-learning method to address the cold-start problem, which can provide a positive initial interaction experience.
We use the Bayesian Linear UCB method to recommend items to users. 
Meanwhile, we provide a theoretical analysis of regret to guarantee its performance.
Finally, extensive experiments on three real-world datasets have demonstrated the remarkable results of our method, which consistently outperforms state-of-the-art baselines.

\section{ Acknowledgments}
The work was supported by grants from the National Key R\&D Program of China (No. 2021ZD0111801) and the National Natural Science Foundation of China (No. 62022077).

\bibliographystyle{elsarticle-num} 
\bibliography{ref}

\appendix
\section{Proof of Results}

\subsection{ Proof of theorem \ref{thm:ucb}}

\begin{lemma}
    \label{lemma:m_info}
    If $\eqbf{\theta} \in \mathbb{R}^d$ follows $\mathcal{N}\left(\eqbf{\mu}, \eqbf{\Sigma}\right)$, and $r = \eqbf{e^{\top}} \eqbf{\theta} + \xi$ where $\eqbf{e} \in \mathbb{R}^d$ is fixed and $\xi \sim \mathcal{N}\left(0, \noise^2\right)$, then
$$
I(\eqbf{\theta};\eqbf{e}, r) = \frac{1}{2} \log \left(1+\frac{\eqbf{e}^{\top} \eqbf{\Sigma} \eqbf{e}}{\noise^2}\right) .
$$
\begin{proof}
We use $h$ to denote the differential entropy of a continuous random variable
$$
\begin{aligned}
I(\eqbf{\theta};\eqbf{e}, r) & =h(\eqbf{\theta})-h(\eqbf{\theta} \mid \eqbf{e},r) \\
& =\frac{1}{2} \log \operatorname{det}(2 \pi \eqbf{e} \eqbf{\Sigma})-\frac{1}{2} \log \operatorname{det}\left(2 \pi \eqbf{e} \left( \eqbf{\Sigma^{-1}} + \frac{\eqbf{e} \eqbf{e^{\top}}}{\noise^2}\right)^{-1}\right) \\
& =\frac{1}{2} \log \operatorname{det}\left(\eqbf{I_d} + \frac{\eqbf{\Sigma} \eqbf{e} \eqbf{e^{\top}}}{\noise^2}\right) \\
& =\frac{1}{2} \log \left(1+\frac{\eqbf{e}^{\top} \eqbf{\Sigma} \eqbf{e}}{\noise^2}\right),
\end{aligned}
$$
where the last step follows from Sylvester's determinant theorem.
\end{proof}
\end{lemma}

\begin{proof}[Proof of theorem \ref{thm:ucb}]
Note that $\hat{r}^{(t)}_{u,i} = \eiTOnline  {\eutOnline}$ is distributed as $\mathcal{N}\left(\eiTOnline  \eqbf{\mu}_{t}, \eiTOnline \eqbf{\Sigma}_{t} \eiOnline\right)$
By the Chernoff bound,
$$
\begin{aligned}
P\left(\left|\hat{r}^{(t)}_{u,i}-\mathbb{E} \hat{r}^{(t)}_{u,i}\right| \geq \frac{\Gamma_{t}}{2} \sqrt{I\left(u;i, r_{u, i}\right)}\right) & \leq 2 \exp \left(-\frac{\left(\frac{\Gamma_{t}}{2} \sqrt{I\left(u;i, r_{u, i}\right)}\right)^2}{2 \eiTOnline \eqbf{\Sigma}_{t} \eiOnline}\right) \\
& =2 \exp \left(-\frac{\Gamma_{t}^2 I\left(u;i, r_{u, i}\right)}{8 \eiTOnline \eqbf{\Sigma}_{t} \eiOnline}\right) \\
& \leq 2 \exp \left(-\frac{\lambda_t}{\log \left(1+\frac{\lambda_t}{\noise^2}\right)} \frac{\log \left(1 + {\eiTOnline \eqbf{\Sigma}_{t} \eiOnline}/{\noise^2}\right)}{\eiTOnline \eqbf{\Sigma}_{t} \eiOnline} \log \frac{2|\mathcal{A}_t|}{\delta}\right) \\
& \leq \frac{\delta}{|\mathcal{A}_t|},
\end{aligned}
$$
where the last inequality follows from the monotonicity of $\frac{x}{\log (1+x)}$ for $x>0$ and the fact that $ \lambda_t = \max _{i \in \mathcal{A}_t} \eiTOnline \eqbf{\Sigma}_{t} \eiOnline  \geq 
 \eiTOnline \eqbf{\Sigma}_{t} \eiOnline  $. Applying a union bound over actions gives
$$
P\left(\left|\hat{r}^{(t)}_{u,i}-\mathbb{E} \hat{r}^{(t)}_{u,i}\right| \leq \frac{\Gamma_{t}}{2} \sqrt{I\left(u;i, r_{u, i}\right)} , \forall i \in \mathcal{A}_t\right) \geq 1 - \delta.
$$
\end{proof}

\subsection{Proof of lemma \ref{lemma:event}}
\begin{proof}[Proof of lemma \ref{lemma:event}]
For brevity, in this part, we denote $\eqbf{\mu_k}$ as ${\Emean g_{u_k}}$.
We analyze the three events $\mathcal{E}_\theta$, $\mathcal{E}_m$, and $\mathcal{E}_s$ separately.

In terms of $\mathcal{E}_\theta$, let $ \eqbf{z} = \eqbf{\Sigma}_*^{-\frac{1}{2}}\left( \euOnline - \eqbf{\mu}_*\right)$, we have $ \eqbf{z} \sim \mathcal{N}\left(0, \eqbf{I_d} \right)$,
$$
\begin{aligned}
P\left(\bar{\mathcal{E}}_\theta\right) 
& = P\left(\exists i \in \left[ d \right], z_i^2 \geqslant 2 \ln \frac{d^2 T}{\delta}\right) \\
& \leq d \cdot P\left(z_1^2 \geqslant 2 \ln \frac{d^2 T}{\delta}\right) \\
& \leq \frac{2 \delta}{d T}. \\
\end{aligned}
$$
The first inequality is due to the countable subadditivity of probability measure, and the last inequality is because $z_1$ is also a $1$-sub-Gaussian distribution, $P(z_1 > t) \leq \exp\left(-\frac{1}{2}t^2\right)$ for $t>0$ holds.
With a probability of at least $1 - \frac{2 \delta}{dT}$, the following inequality holds:
\begin{equation}
\label{eq::loc::e_theta}
\left\|\eqbf{\Sigma}_*^{-\frac{1}{2}}\left(\euOnline-\eqbf{\mu}_*\right)\right\|_{\infty}^2 \leqslant 2 \ln \left(\frac{d^2 T}{\delta}\right).
\end{equation}

In terms of $\mathcal{E}_m$, for any unit vector $\|\eqbf{\nu}\| = 1, s\in \mathbb{R} $,  we have
\begin{equation}
\label{eq:loc:sg}
\begin{aligned}
\mathbb{E}\left[ s \cdot \eqbf{\nu}^\top({\eqbf{\mu}_{\text {meta }} - \eqbf{\mu}^*}) \right] 
& = 
\mathbb{E}\left[
\exp \left( \frac{s}{M} \eqbf{\nu^\top} \sum_{i=1}^M \left( {\eqbf{\mu}_i - \eqbf{\mu}}^*  \right) \right) \right] \\
&= \prod_{i=1}^M \mathbb{E}\left[exp\left({\frac{s}{M} \cdot \eqbf{\nu^{\top}}\left({\eqbf{\mu}_i - \eqbf{\mu}}^*\right)}\right)\right] \\
&  \underset{(a)}{=} \prod_{i=1}^M \exp \left\{ \frac{1}{M^2} \cdot \eqbf{\nu^{\top}} \eqbf{\Sigma}_* \eqbf{\nu} \cdot \frac{s^2}{2}\right\} \\
& =\exp \left\{\frac{s^2}{2} \frac{\eqbf{\nu^{\top}} \eqbf{\Sigma}_* \eqbf{\nu}}{M}\right\} \\
& \underset{(b)}{\leqslant} \exp \left\{\frac{s^2}{2} \cdot \frac{\Bar{\lambda}}{M}\right\}, \\
\end{aligned}
\end{equation}
where (a) uses the the MGF of a Gaussian distribution with 
$
\frac{1}{M} \eqbf{\nu^{\top}\left({\eqbf{\mu}_i - \eqbf{\mu}}^*\right)} \sim \mathcal{N}\left(0, \frac{1}{M^2} \cdot \eqbf{\nu^{\top} \eqbf{\Sigma}_* \nu}\right)
$,
(b) is because $\Bar{\lambda}$ is the largest eigenvalue of $\eqbf{\Sigma}_{*}$.

From above equation \eqref{eq:loc:sg}, we get $\eqbf{\mu_{\text {meta }}-\mu^*} \sim \operatorname{subG} \left[\sqrt{\frac{\Bar{\lambda}}{M}}\right]$.
By lemma \ref{lemma::aux::subG}, select A as identity matrix in lemma \ref{lemma::aux::subG}, we have
$$
\mathbb{P}\left(\|\eqbf{\mu}_\text{meta} - \eqbf{\mu}_*\|^2  > \frac{\Bar{\lambda}}{M} \left(d + 2\sqrt{d \ln{\frac{1}{\delta}}} + 2 \ln{\frac{1}{\delta}}\right)\right) \leq \delta.
$$
With a probability of at least $1 - \frac{\delta}{dT}$, the following inequality holds:
\begin{equation}
\label{eq:mean_assign}
\|\eqbf{\mu}_\text{meta} - \eqbf{\mu}_*\|  \leq \sqrt{\frac{\Bar{\lambda}}{M}\left( 2d + 3\ln{\frac{d T}{\delta}} \right)}.
\end{equation}

In terms of $\mathcal{E}_s$, as $\eqbf{\mu_i} - \eqbf{\mu}_\text{meta} \sim \mathcal{N}\left(0, \frac{M-1}{M} \eqbf{\Sigma}_* \right)$, we have $\eqbf{\mu_i} - \eqbf{\mu}_\text{meta} \sim subG\left(\sqrt{\Bar{\lambda}}\right)$, 
by lemma \ref{lemma::aux::covariance}, with a probability of at least $1- \frac{6\delta}{dT}$, the following inequality holds:
\begin{equation}
\|{\boldsymbol{\Sigma}}_\text{meta} - \boldsymbol{\Sigma}_*\|_{\mathrm{op}} \leq 
32 \Bar{\lambda} \cdot \max \left\{\sqrt{\frac{5d + 2 \ln \left(\frac{dT}{3\delta}\right)}{M}}, \frac{5 d+2 \ln \left(\frac{dT}{3\delta}\right)}{M}\right\}.
\end{equation}
Choose $\delta = \frac{1}{M}$,  for $M \geq 5d + 2 \ln \left(\frac{dMT}{3}\right)$,  by lemma \ref{lemma::aux::wideop}, we have 
\begin{equation}
\label{eq:cov_assign}
\|\AlgPriorCov-\boldsymbol{\Sigma}_*\|_{\mathrm{op}} \leq 
64 \Bar{\lambda} \cdot \sqrt{\frac{5d + 2 \ln \left(\frac{dMT}{3}\right)}{M}},
\quad \AlgPriorCov > \boldsymbol{\Sigma}_*.
\end{equation}

By combining equations \eqref{eq::loc::e_theta}, \eqref{eq:mean_assign}, and \eqref{eq:cov_assign}, set 
\begin{equation}
    \label{eq:fms}
    \begin{aligned}
& \delta = 1/M, \\
& f_m =  {{\Bar{\lambda}}\left(2d + 3\ln{{\left(d M T\right)}} \right)}, \\
& f_s =  (64 \Bar{\lambda})^2 \left({{5d + 2 \ln \left(\frac{dMT}{3}\right)}}\right), \\
\end{aligned}
\end{equation}
we can obtain the desired result.

\end{proof}

\subsection{Proof of theorem \ref{thm:regret}}
In previous work \cite{peleg2022metalearning}, the algorithm setting is random exploration of initial $\tau$ rounds.
This is a major difference between the algorithm we use in online and the existed method. 
Here, we first decompose the regret into two stages: the first $\tau$ rounds and the remaining $T - \tau$ rounds. We will analyze these two terms separately. 
After establishing the relationship between the regret of inaccurate priors and correct priors, we can then combine this with the standard conclusions of Bayesian bandit regret to obtain the main results.

For brevity, in this part, we denote $ \mathcal{N}(\eqbf{{\mu}}_0, \AlgPriorCov)$ 
as algorithm prior and $\mathcal{N}(\eqbf{{\mu}}_t, \eqbf{{\Sigma}}_t)$ represents the posterior in round $t$ when the algorithm is initialized with prior $\mathcal{N}(\eqbf{{\mu}}_0, \AlgPriorCov)$, and $\mathcal{N}(\eqbf{{\mu}}_{*,t}, \eqbf{{\Sigma}}_{*, t})$ represents the posterior in round $t$ when the algorithm is initialized with the correct prior $\mathcal{N}(\eqbf{{\mu}}_{*}, \eqbf{{\Sigma}}_{*})$.

\begin{proof}[Proof of theorem \ref{thm:regret}]

We first decompose the regret into two stages: the initial $\tau$ rounds and the remaining $T - \tau$ rounds.

\begin{equation}
\label{eq:RT tau}
\begin{aligned}
\text{Reg}_{\text{Bay}}(\eqbf{{\mu}}_0, \AlgPriorCov, T) &= \text{Reg}_{\text{Bay}}(\eqbf{{\mu}}_0, \AlgPriorCov, \tau) + \text{Reg}_{\text{Bay}}(\eqbf{{\mu}}_{\tau + 1}, \AlgPriorCovTau, T - \tau). \\    
\end{aligned}
\end{equation}

For the initial $\tau$ rounds, intuitively, as we have not yet observed a sufficient amount of interaction data, we make the following worst-case estimation for this part.

\begin{equation}
\label{eq:tau1}
\begin{aligned}
\text{Reg}_{\text{Bay}}(\eqbf{{\mu}}_0, \AlgPriorCov, \tau)  
& \underset{(a)}{\leq} {\mathbb{E}}_{\euOnline}\left[\sum_{t=1}^\tau \max _{a \in \mathcal{A}_t}\left\{\eqbf{{e}}_{i_a}^{* \top} \euOnline\right\} \right]-{\mathbb{E}}_{\euOnline}\left[\sum_{t=1}^\tau \min _{a \in \mathcal{A}_t}\left\{\eqbf{{e}}_{i_a}^{* \top} \euOnline\right\} \right] \\
& \leq 2 \mathbb{E}_{\euOnline}\left[\sum_{t=1}^\tau \max _{a \in \mathcal{A}_t}\left\{\left|\eqbf{{e}}_{i_a}^{* \top} \euOnline\right|\right\} \right] \\
& \underset{(b)}{\leq} 2 \mathbb{E}_{\euOnline}\left[\sum_{t=1}^\tau \max _{a \in \mathcal{A}_t}\{\| \eqbf{{e}}^{*}_{i_a} \| \cdot \|\euOnline\|\} \right] \\
& \underset{(c)}{\leq} 2 a \tau {\mathbb{E}}_{\euOnline}\left[\|\euOnline\| \right],
\end{aligned}
\end{equation}
where $(a)$ is the maximal regret of any algorithm, $(b)$ uses Cauchy-schwarz inequality and $(c)$ uses Assumption \ref{ass:bound}. Denote $\eqbf{Z} \triangleq \eqbf{\Sigma}_*^{-1 / 2}\left(\euOnline-\eqbf{\mu}_*\right)$ and analyzing the expectation,

\begin{equation}
    \label{eq:tau2}
\begin{aligned}
\mathbb{E}\left[\|\euOnline\| \right] & \underset{(a)}{\leq}\left\|\eqbf{\mu}_*\right\|+\mathbb{E}\left[\left\|\euOnline-\eqbf{\mu}_*\right\| \right] \\
& \underset{(b)}{\leq} m+\sqrt{\bar{\lambda}} \mathbb{E}\left[\|\eqbf{Z}\| \right] \\
& =m+\sqrt{\bar{\lambda}} \mathbb{E}\left[\sqrt{\sum_{i=1}^d Z_i^2} \right] \\
& \underset{(c)}{\leq} m+\sqrt{\bar{\lambda}} \sqrt{\sum_{i=1}^d \mathbb{E}\left[Z_i^2 \right],}\\
& = m + \sqrt{\bar{\lambda} d},  \\
\end{aligned}
\end{equation}
where $(a)$ uses the triangle inequality, $(b)$ uses Lemma \ref{lemma::aux::max_eige}, and $(c)$ uses Jensen inequality.
Combine equation \eqref{eq:tau1} and \eqref{eq:tau2}, we can get 
\begin{equation}
    \label{eq:tau regret}
\text{Reg}_{\text{Bay}}(\eqbf{{\mu}}_0, \AlgPriorCov, \tau) \leq k_2\tau,
\end{equation}    
where $k_2 = 2a \left(m + \sqrt{\bar{\lambda} d}\right)$.

For the second part, combining the definition \ref{def:tau} of $\tau$ and lemma \ref{lemma:event}, we can refer to previous results to provide a connection between the regret of inaccurate priors and correct priors. 

\begin{equation}
\label{eq:T-tau}
\begin{aligned}
\text{Reg}_{\text{Bay}}(\eqbf{{\mu}}_{\tau + 1}, \AlgPriorCovTau, T - \tau) 
& \underset{(a)}{\leq} (1 + k_1) \text{Reg}_{\text{Bay}}({\eqbf{\mu}}_{*,\tau + 1}, \eqbf{\Sigma}_{*,\tau + 1}, T - \tau) + \frac{c_{bad}\delta}{\sqrt{d}}\\
& \leq (1 + k_1) \text{Reg}_{\text{Bay}}({\eqbf{\mu}_{*}}, \eqbf{\Sigma}_{*}, T) + \frac{c_{bad}\delta}{\sqrt{d}}\\
& \underset{(b)}{\leq} (1 + k_1) \left( \Gamma \sqrt{ \frac{1}{2}Td \log \left(1+\frac{\Bar{\lambda}T}{\noise^2}\right)} + B\right)  + \frac{c_{bad}\delta}{\sqrt{d}},\\
\end{aligned}
\end{equation}
where  (a) uses lemma \ref{lemma:event} and lemma \ref{lemma:meta}, and (b) uses Lemma \ref{lemma:bayesian regret}.
By combining equations \eqref{eq:RT tau}, \eqref{eq:tau regret}, and \eqref{eq:T-tau}, we can obtain the desired result.

\end{proof}

\begin{lemma}[Theorem 1 in \cite{peleg2022metalearning}]
\label{lemma:meta} 
Let $\euOnline \sim \mathcal{N}\left(\eqbf{\mu}_*, \boldsymbol{\Sigma}_*\right)$ and let $\mathcal{N}(\eqbf{{\mu}}_0, {\boldsymbol{\Sigma}_0})$ be the prior. For $\tau<T$, if for some $0<\delta \leq 1 / c_{\delta}$ the event $\mathcal{E}$ holds with probability larger than $1 - \frac{9 \delta}{d T}$, then the regret is bounded by,

\begin{equation}
    \text{Reg}_{\text{Bay}}(\eqbf{{\mu}}_{\tau + 1}, {\Sigma}_{\tau + 1}, T - \tau) 
    \leq (1 + k_1) \text{Reg}_{\text{Bay}}({\eqbf{\mu}}_{*,\tau + 1}, \eqbf{\Sigma}_{*,\tau + 1}, T - \tau) + \frac{c_{bad}\delta}{\sqrt{d}},
\end{equation}
where 
$$
\begin{aligned}
& c_{\delta}=\max \left\{3, c_s^2 \tau^2 f_s, 18 c_{\xi}^2 c_s\left(f_m+\left(c_1 d+c_{\xi}^2 c_s / 36\right) f_s\right)\right\}, \\
& k_1=12 \sqrt{c_{\xi}^2 c_s} \sqrt{f_m \delta}+\left(c_s \tau+12 \sqrt{c_{\xi}^2 c_s c_1 d}+2 c_{\xi}^2 c_s\right) \sqrt{f_s \delta}, \\
& c_s=\frac{2 \noise^2}{\underline{\lambda}^2 \lambda_{\boldsymbol{\Sigma}_{{A}}}}, \quad c_{\xi}=\sigma \sqrt{5 \ln \left(\frac{d T}{\delta}\right)}, \quad c_1=\frac{2}{\underline{\lambda}} \ln \left(\frac{d^2 T}{\delta}\right), \quad c_{\text {bad }}=22 a\left(m+\sqrt{4 \bar{\lambda} \ln \left(\frac{d^2 T}{\delta}\right)}\right) .
\end{aligned}
$$

\end{lemma}

\begin{lemma}[Proposition 6 and Lemma 7 in \cite{lu2019information}]
\label{lemma:bayesian regret}
Under the assumptions and notation in section \ref{sec:thm}, 
the Bayesian regret of UCB over T periods is
$$
\text{Reg}_{\text{Bay}}({\eqbf{\mu}_{*}}, \eqbf{\Sigma}_{*}, T) 
\leq \Gamma \sqrt{ \frac{1}{2}Td \log \left(1+\frac{\Bar{\lambda}T}{\noise^2}\right)} + B, 
$$
where $ B = a \left(m + \sqrt{\Bar{\lambda} d } \right) $ and
$ \Gamma = 4 \sqrt{\frac{\Bar{\lambda}}{\log \left(1+\frac{\Bar{\lambda}}{\noise^2}\right)} \log (4 N T)}.$    
\end{lemma}

\section{Auxiliary Lemmas}

\begin{lemma}[Concentration bound sub-Gaussian vector, Theorem 2.1 in \cite{hsu2012tail}]
\label{lemma::aux::subG}
Let $\eqbf{A} \in \mathbb{R}^{m \times n}$ be a matrix, and let $\boldsymbol{\Sigma}=\eqbf{A}^{\top} \eqbf{A}$. Suppose that $X=\left(X_1, \ldots, X_n\right)$ is a random vector such that $ \mathbb{E}[X] = 0$ and for some $\sigma \geq 0$
$$
\mathbb{E}\left[\exp \left(U^{\top} X\right)\right] \leq \exp \left(\frac{\|U\|^2 \sigma^2}{2}\right),
$$
for all $U \in \mathbb{R}^n$. For all $\delta>0$,
$$
\mathbb{P}\left(\|\eqbf{A} X\|^2>\sigma^2\left(\operatorname{Tr}(\boldsymbol{\Sigma})+2 \sqrt{\operatorname{Tr}\left(\boldsymbol{\Sigma}^2\right) \delta}+2\|\boldsymbol{\Sigma}\|_{\mathrm{op}} \delta\right)\right) \leq e^{-\delta}.
$$

\end{lemma}

\begin{lemma}[Empirical covariance bounds, Theorem 6.5 in \cite{wainwright2019high}, constants were taken from \cite{bastani2022meta}]
\label{lemma::aux::covariance}
For any row-wise $\sigma$ sub-Gaussian random matrix $\mathbf{X} \in \mathbb{R}^{n \times d}$, the sample covariance matrix $\hat{\boldsymbol{\Sigma}} =\frac{1}{n} \sum_{i=1}^n X_i X_i^{\top}$ satisfies the bound, $\forall 0<\delta<1$, 
$$
\mathbb{P}\left(\|\hat{\boldsymbol{\Sigma}}-\boldsymbol{\Sigma}\|_{\mathrm{op}} \geq 32 \sigma^2 \cdot \max \left\{\sqrt{\frac{5 d+2 \ln \left(\frac{2}{\delta}\right)}{n}}, \frac{5 d+2 \ln \left(\frac{2}{\delta}\right)}{n}\right\}\right) \leq \delta .
$$
\end{lemma}

\begin{lemma}[Maximal eigenvalue inequality, lemma 26 in \cite{peleg2022metalearning}]
\label{lemma::aux::max_eige}
Let $\eqbf{v}$ be a vector and $\mathbf{B}$ a positive definite
matrix, then,
$$
\|\eqbf{v}\| \leq \sqrt{\lambda_{\max }(\mathbf{B})}\|\eqbf{v}\|_{\mathbf{B}^{-1}}.
$$
\end{lemma}

\begin{lemma}
\label{lemma::aux::wideop}
Let $\hat{\boldsymbol{\Sigma}}$ be a symmetric matrix and $\boldsymbol{\Sigma}_*$ be a PD matrix s.t. $\left\|\hat{\boldsymbol{\Sigma}}-\boldsymbol{\Sigma}_*\right\|_{\mathrm{op}} \leq$ s. Define $\hat{\boldsymbol{\Sigma}}^{\mathrm{w}}=\hat{\boldsymbol{\Sigma}}+s \cdot I$, then
$$
\hat{\mathbf{\Sigma}}^{\mathrm{w}} \succeq \boldsymbol{\Sigma}_*.
$$
\begin{proof}    
$$
\lambda_{\min }\left(\hat{\boldsymbol{\Sigma}}^{\mathrm{w}}-\boldsymbol{\Sigma}_*\right) \underset{(a)}{\geq} \lambda_{\min }(s \cdot \mathbf{I})+\lambda_{\min }\left(\hat{\boldsymbol{\Sigma}}-\boldsymbol{\Sigma}_*\right) \underset{(b)}{\geq} s-\left\|\boldsymbol{\Sigma}_*-\hat{\boldsymbol{\Sigma}}\right\|_{\mathrm{op}} \geq 0,
$$
where $(a)$ uses Weyl's inequality and $(b)$ uses $\lambda_{\min }\left(\hat{\boldsymbol{\Sigma}}-\boldsymbol{\Sigma}_*\right)=-\lambda_{\max }\left(\boldsymbol{\Sigma}_*-\hat{\boldsymbol{\Sigma}}\right) \geq-\left\|\boldsymbol{\Sigma}_*-\hat{\boldsymbol{\Sigma}}\right\|_{\text {op }}$. 
\end{proof}
\end{lemma}

\end{document}